\documentclass[smallextended]{svjour3}
\usepackage{tikz}
\usepackage{pgf}
\usetikzlibrary{arrows,chains,automata,positioning,calc,decorations.pathreplacing,shapes,snakes,fit}
\usepackage{verbatim}
\usepackage{multirow}
\usepackage{graphicx}
\usepackage[ruled,lined,linesnumbered,boxed,vlined]{algorithm2e}
\usepackage
{amsmath, amssymb}
\usepackage{amsthm} 
\usepackage{hyperref}
\usepackage[hyphenbreaks]{breakurl}
\usepackage[tight,footnotesize]{subfigure}
\usepackage{listings}
\usepackage{framed,color}
\usepackage{calc}
\usepackage{array}
\usepackage{float}
\usepackage[squaren]{SIunits}
\usepackage{sistyle}

\definecolor{gray}{rgb}{0.5,0.5,0.5}
\definecolor{shadecolor}{rgb}{0.5,0.5,0.5}

\makeatletter
\protected\def\includeGraphics{\@testopt\roy@includegraphics{}}
\def\roy@includegraphics[#1]#2{
    \begingroup
        \edef\x{\endgroup\noexpand\includegraphics[#1]}\x{#2}
}
\makeatother

\tikzset{
    state/.style={
        circle,
        draw=black, very thick,
        minimum size = 1cm,
    },
    smallstate/.style={
        circle,
        draw=black, very thick,
        minimum size = 3mm,
        inner sep=0.8mm,
    },
    smallfinalstate/.style={
        circle,
        double,
        draw=black, very thick,
        minimum size = 3mm,
        inner sep=0.8mm,
    },
    description/.style={
        rectangle,
        draw=white,
        minimum height = 2em,
        inner sep = 0pt,
    },
    tuborg/.style={
        decorate
    },
}

\newcommand\States{\ensuremath{Q}}
\newcommand\FinalStates{\ensuremath{F}}
\newcommand\NrStates{\ensuremath{\lvert\States\rvert}}
\newcommand\State[1][{}]{\ensuremath{q_{#1}}}
\newcommand\StartState{\ensuremath{\State[0]}}
\newcommand\ErrorState{\ensuremath{\State[e]}}
\newcommand\Processors{\ensuremath{P}}
\newcommand\NrProcessors{\ensuremath{\lvert\Processors\rvert}}
\newcommand\Processor[1][{}]{\ensuremath{p_{#1}}}
\newcommand\Cores{\ensuremath{C}}
\newcommand\NrCores{\ensuremath{\lvert\Cores\rvert}}
\newcommand\Chunk[1][{}]{\ensuremath{c_{#1}}}

\newcommand\Length[1][{}]{\ensuremath{L_{#1}}}
\newcommand\Weight[1][{}]{\ensuremath{{w_{#1}}}}
\newcommand\Matched[1][{}]{\ensuremath{{m_{#1}}}}
\newcommand\StateMap[1][{}]{\ensuremath{{\mathcal{L}_{#1}}}}
\newcommand\IStates[1][{}]{\ensuremath{{\mathcal{I}_{#1}}}}
\newcommand\MaxNrIStates{\ensuremath{{\mathcal{I}_{\text{max}}}}}
\newcommand\TF{\ensuremath{\delta}}
\newcommand\TFE{\ensuremath{\TF^*}}
\newcommand\Sym{\ensuremath{\sigma}}
\newcommand\STR{\textit{Str}}
\newcommand\BigO{\ensuremath{\mathcal{O}}}
\newcommand\falset{\text{\em false}}
\newcommand\truet{\text{\em true}}

\DeclareMathOperator{\StartPos}{StartPos}
\DeclareMathOperator{\EndPos}{EndPos}

\titlerunning{Speculative DFA Membership Tests}

\newcommand\NrPCREs{299}
\newcommand\NrPROSITEs{110}

\begin{document}

\title{A Speculative Parallel DFA Membership Test for Multicore, SIMD and Cloud
Computing Environments}

\author{Yousun Ko
        \and Minyoung Jung
        \and \makebox{Yo-Sub Han}
        \and Bernd Burgstaller}
\institute{Yousun Ko
        \and Minyoung Jung
        \and Yo-Sub Han
        \and Bernd Burgstaller
        \at Department of Computer Science,
        Yonsei University, Seoul, Korea
        \\\email{bburg@cs.yonsei.ac.kr}
}

\date{}

\maketitle

\begin{abstract}
We present techniques to parallelize membership tests for
Deterministic Finite Automata (DFAs). Our method searches
arbitrary regular expressions by matching multiple
bytes in parallel using speculation. We partition the input
string into chunks, match chunks in parallel, and combine the 
matching results. Our parallel
matching algorithm exploits structural DFA properties to minimize
the speculative overhead. Unlike previous approaches, our
speculation is {\em failure-free\/}, i.e., (1) sequential
semantics are maintained, and (2) speed-downs are avoided
altogether.
On architectures with a SIMD gather-operation for indexed memory loads,
our matching operation is fully vectorized.
The proposed load-balancing scheme uses
an off-line profiling step to determine the matching capacity of each
participating processor. Based on matching capacities, DFA matches
are load-balanced on inhomogeneous parallel architectures such as
cloud computing environments.

We evaluated our speculative DFA membership test for a representative 
set of benchmarks from the Perl-compatible
Regular Expression (PCRE) library~\cite{PCRELib} and
the PROSITE~\cite{PROSITE} protein database. Evaluation was conducted
on a 4~CPU (40~cores) shared-memory node of the Intel
Academic Program Manycore Testing Lab (Intel MTL), 
on the Intel AVX2 SDE simulator for 8-way 
fully vectorized SIMD execution, and on a 20-node (288~cores)
cluster on the Amazon EC2 computing cloud.
Obtained speedups are on the order of
$\BigO(1+\frac{\lvert P\rvert-1}{\lvert Q\rvert\cdot\gamma})$,
where $\lvert P\rvert$ denotes the number of processors or SIMD units, 
$ \lvert Q\rvert$ denotes the number of DFA states, and $0<\gamma\le 1$ 
represents a statically computed DFA property. For all 
observed cases, we found that $0.02<\gamma<0.47$. Actual speedups range 
from 2.3x to 38.8x for up to 512~DFA states for PCRE, and between 1.3x and 19.9x 
for up to 1288~DFA states for PROSITE on a 40-core MTL node. 
Speedups
on the EC2 computing cloud range from 5.0x to 65.8x for PCRE, and from 5.0x to 138.5x
for PROSITE. Speedups of our C-based DFA matcher over the Perl-based
ScanProsite scan tool~\cite{ScanProsite} range from 559.3x to 15079.7x on a 40-core 
MTL node. We show the scalability of our approach for input-sizes of up to \SI{10}{GB}.

\end{abstract}

\section{Introduction}\label{sec:Introduction}
Locating a string within a larger text has applications
with text editing, compiler front-ends and web browsers,
scripting languages, file-search (grep), command-processors, databases,
Internet search engines, computer security, and DNA sequence analysis.
Regular expressions allow the specification of a potentially infinite
set of strings (or patterns) to search for.
A standard technique to perform regular expression matching is to
convert a regular expression to a DFA and run the DFA on
the input text. DFA-based regular expression matching has robust, linear
performance in the size of the input. However, practical DFA implementations
are inherently sequential as the matching result of an input character is 
dependent on 
the matching result of the previous characters.
Related to DFA matching on parallel architectures, considerable
research effort has been recently
spent~\cite{Luchaup2011,WangHL10,Holub:2009,Jones2009,Luchaup2009,ScarpazzaVP07}.

To speed up DFA matching on parallel architectures, we propose
to use speculation. With our method, the input string is divided
into chunks. Chunks are processed in parallel using sequential
DFA matching. For all but the first chunk, the starting state is unknown.

The core contribution of our method is to exploit structural properties
of DFAs to bound the set of initial states the DFA may assume
at the beginning of each chunk. Each chunk will be matched
for its reduced set of possible
initial states. By introducing such a limited amount
of redundant matching computation for all but the first chunk,
our DFA matching algorithm avoids speed-downs altogether (i.e.,
the speculation is failure-free~\cite{Specpar}). To achieve load-balancing,
the input string is partitioned
non-uniformly according to processor
capacity and work to be performed for each chunk.
These properties open up the opportunity for an entire new class
of parallel DFA matching algorithms. 
We present the time complexity of our matching algorithms,
and we conduct an extensive experimental evaluation
on SIMD, shared-memory multicore and cloud computing environments.
For experiments, we employ regular expressions from 
the PCRE Library~\cite{PCRELib} and from the
PROSITE protein pattern database~\cite{PROSITE}.
We show the scalability of our approach for input-sizes of
up to \SI{10}{GB}.

The paper is organized as follows. In
Section~\ref{sec:Background}, we introduce background material.
In Section~\ref{sec:Overview}, we discuss a motivating example
for our speculative DFA matching algorithms.
In Section~\ref{sec:Algorithms}, we introduce our algorithms and
their complexity with respect to speedup and costs.
Section~\ref{sec:Implementation} shows three implementations for
SIMD, shared-memory multicore and cloud-computing
environments. Section~\ref{sec:ExperimentalResults}
contains experimental results.
We discuss the related work in Section~\ref{sec:RelatedWork}
and draw our conclusions in
Section~\ref{sec:Conclusion}.

\begin{algorithm}[t]
    \SetKwInOut{Input}{Input}
    \SetKwInOut{Output}{Output}
    \SetKwData{STATE}{state}
    \SetKw{RETURN}{return}
    \Input{transition function~\TF, input string~$\STR$,
    start state~\StartState, set of final states~\FinalStates}
    \Output{$\truet$\/ if input is matched, $\falset$ otherwise}
    \STATE$\leftarrow\StartState$\\
    \For{$i\leftarrow0$ \KwTo $\lvert\STR\rvert-1$}{
        \STATE$\leftarrow\TF(\STATE,\STR[i])$
    }
    \If{\STATE$\in\FinalStates$}{
       \RETURN $\truet$\tcp*{input matched}
    }
    \RETURN $\falset$
    \caption{Sequential DFA matching}
    \label{algo:seqmatch}
\end{algorithm}

\section{Background\label{sec:Background}}
\subsection{Finite Automata\label{sec:FiniteAutomata}}
Let~$\Sigma$ denote a finite alphabet of characters and $\Sigma^*$
denote the set of all strings over $\Sigma$. Cardinality~$\vert\Sigma\vert$
denotes the number of characters in $\Sigma$.  A language over
$\Sigma$ is any subset of $\Sigma^*$. The symbol $\emptyset$ denotes
the empty language and the symbol $\lambda$ denotes the null string.
A finite automaton~$A$ is specified by a tuple 
$(\States, \Sigma, \TF, \StartState,F)$, 
where $\States$ is a finite set of states, $\Sigma$ is an input alphabet, 
$\TF: \States \times \Sigma \to 2^\States$ is a transition function,
$\StartState\in\States$ is the start state and $F \subseteq \States$ 
is a set of final states.  We define $A$ to be a DFA if $\TF$ is a 
transition function of $\States\times \Sigma \to \States$
and $\TF(\State, a)$ is a singleton set for any $\State \in \States$ and 
$a\in \Sigma$. 
Let $|\States|$ be the number of states in $\States$.
We extend transition function~$\TF$ to $\TFE$:
$\TFE(\State, ua)=p\Leftrightarrow\TFE(\State,u)=\State'$, 
$\TF(\State',a)=p$, $a\in\Sigma$, $u\in\Sigma^*$. 
We assume that a DFA has a unique error (or sink) state~$\ErrorState$.

An input string~$\STR$ over $\Sigma$ is accepted by DFA~$A$ if the DFA contains a labeled
path from $\StartState$ to a final state such that this path reads~$\STR$.
We call
this path an accepting path. Then, the language~$L(A)$ of $A$ is the
set of all strings spelled out by accepting paths in~$A$.

The DFA membership test determines whether a string is 
contained in the language of a DFA.
The DFA membership test is conducted by
computing~$\TFE(\State[0],\STR)$ and checking whether the result 
is a final state. Algorithm~\ref{algo:seqmatch} denotes the
sequential DFA matching algorithm.
As a notational convention, we denote the symbol
in the ~$i\,$th position of the input string by~$\STR[i]$.

\subsection{Amazon EC2 Infrastructure\label{sec:ec2background}}
The Amazon Elastic Computing Cloud (EC2) allows users to rent virtual
computing nodes on which to run applications. EC2 is very popular
among researchers and companies in need of instant and scalable computing
power. Amazon EC2 provides resizable compute capacity where users 
pay on an hourly basis for launched (i.e., up-and-running) virtual nodes.
By using virtualized resources, a computing cloud can serve a much broader
user base with the same set of physical resources.
Amazon EC2 virtual computing
nodes are virtual machines running on top of a 
variant of the Xen hypervisor.
To create a virtual machine, EC2 provides machine images which contain
a pre-configured
operating system plus application software. Users can adapt
machine images prior to deployment. The launch of a machine
image creates a so-called instance, which is a copy of the machine
image executing as a  virtual server in the cloud.
To provide a unit of measure for the compute capacities of instances,
Amazon introduced so-called EC2 Compute Units (CUs), which are claimed
to provide the equivalent CPU capacity of a 1.0--1.2~GHz 2007 Opteron or
2007 Xeon processor~\cite{EC2}. Because there exist many such CPU models
in the market, the exact processor capacity equivalent to one CU is not
entirely clear. Instance types are grouped into nine families, which
differ in their processor, I/O, memory and network capacities.
Instances are described in~\cite{EC2}; the instances employed in this
paper are outlined in Section~\ref{sec:Implementation}.
To create a cluster of EC2 instances, the user requires the launch of
one or more instances, for which the instance type and the machine image
must be specified.
The user can specify any machine image that has been registered with Amazon,
including Amazon's or the user's own images.
Once instances are booted, they are accessible
as computing nodes via ssh. A maximum of 20~instances can be used
concurrently, but this limit may be increased upon
user request~\cite{EC2limit}.  
\section{Overview}\label{sec:Overview}

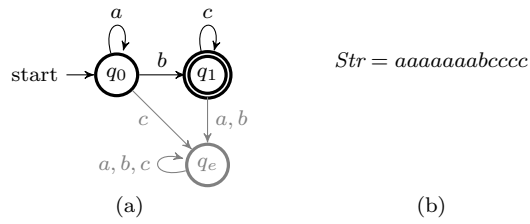
\begin{figure}[t]
    \centering
    \begin{tabular}{ccc}
    \begin{tikzpicture}[->,>=stealth']
        \tikzstyle{every initial by arrow}=[->]
        \node[smallstate, initial] (S0) {\StartState};
        \node[smallfinalstate, right of = S0, node distance = 4em] 
            (S1){\State[1]};
        \node[smallstate, below of = S1,
              node distance = 4em, xshift=0mm, color=gray]
            (SE){\ErrorState};
        \path   (S0)    edge[loop above]    node{$a$}   ()
                (S0)    edge                node[above]{$b$}   (S1)
                (S1)    edge[loop above]    node{$c$}   ()
        ;
        \path   (S0)    edge[color=gray]    node[left,xshift=-0.8mm]{$c$}   (SE)
                (S1)    edge[color=gray]    node[right]{$a,b$} (SE)
                (SE)    edge[loop left, color=gray]     node{$a,b,c$} (SE)
        ;
    \end{tikzpicture}&
    \makebox{\hspace{2mm}}&
    \raisebox{16mm}{$\STR=aaaaaaabcccc$}\\
    {\small (a)}&
    \makebox{\hspace{2mm}}&
    {\small (b)}\\
    \end{tabular}
    \caption{Example DFA including the error state~\ErrorState~(a) and 12-symbol input string~(b).}
    \label{fig:motivDFA}
\end{figure}

The core idea behind our speculative DFA matching method is to divide the input
into several chunks and process them in parallel. As a motivating example we
consider the DFA depicted in
Fig.~\ref{fig:motivDFA}.
This DFA accepts strings which 
contain zero or more occurrences of the symbol~$a$, followed by exactly one occurrence of symbol~$b$,
followed by zero
or more occurrences of symbol~$c$. For the exposition of this motivating example
we have included the DFA's error state~\ErrorState\
and its adjacent transitions, which are depicted in gray. The DFA's alphabet is $\Sigma=\{a,b,c\}$, and
we consider the 12-symbol input string from Fig.~\ref{fig:motivDFA}(b).

Assuming that it takes on the order of one time-unit to process one character from the input string,
Algorithm~\ref{algo:seqmatch} will spend  12~time units for the sequential
membership test. 
This is denoted by the following notation, were a processor~\Processor[0]
matches the input string from Fig.~\ref{fig:motivDFA}(b). The
DFA is in state~\StartState\ initially.

\def\iskip{0.8mm}
\def\cskip{1.8mm}
\newlength{\StrH}
\settoheight{\StrH}{$b$}
\def\colwidth{14mm}

\begin{figure}[H]
 \centering%
 \begin{tabular}{l}
    \rule{\colwidth}{0pt}\\[-0.2mm]
    \fbox{$a\,a\,a\,a\,a\,a\,a\,b\,c\,c\,c\,c$}\\[\iskip]
    \Processor[0]: \StartState
 \end{tabular}
\caption{Notation: example input matched by
processor~\Processor[0]
starting in DFA state~\StartState\label{fig:input_string}.}
\end{figure}

To parallelize the membership test for three processors, the input 
string from Fig.~\ref{fig:input_string}
can be partitioned into three chunks of four symbols each, and assigned
to processors~\Processor[0], \Processor[1] and~\Processor[2] as follows.

\setlength{\intextsep}{2mm}
\begin{figure}[H]
 \centering%
 \begin{tabular}{lll}
    \rule{\colwidth}{0pt}&\rule{\colwidth}{0pt}&\rule{\colwidth}{0pt}\\[-2.2mm]
    \Chunk[0]&\Chunk[1]&\Chunk[2]\\
    \fbox{$a\,a\,a\,a$\rule{0pt}{\StrH}}&
    \fbox{$a\,a\,a\,b$}&
    \fbox{$c\,c\,c\,c$\rule{0pt}{\StrH}}\\[\iskip]
    \Processor[0]: \StartState&
    $\Processor[1]$: $\StartState, \State[1]$&
    $\Processor[2]$: $\StartState, \State[1]$
 \end{tabular}
 \caption{Dividing the input into equal-sized chunks.\label{fig:naive}}
\end{figure}
Because the DFA will initially
be in
start state~\StartState,
the first chunk~(\Chunk[0]) needs to be matched
for \StartState\ only. 
For all subsequent chunks, the DFA state at the beginning of the chunk is initially unknown.
Hence, we use speculative computations to match subsequent chunks for all states the DFA may assume.
We will discuss in Section~\ref{sec:Algorithms} how the amount of speculative computations
can be kept to a minimum.
For our motivating example, we assume the DFA to be in either
state~\State[0] or~\State[1] at the beginning of chunks~\Chunk[1]
and~\Chunk[2]. As depicted by the partition from Fig.~\ref{fig:naive},
processor~\Processor[0] will match chunk~\Chunk[0] for state~\StartState,
whereas processors~\Processor[1] and~\Processor[2] will match their assigned
chunks for both~\State[0] and~\State[1].
To match a chunk for a given state, a variation of the matching loop (lines~1--3)
of Algorithm~\ref{algo:seqmatch} is employed.

After processors~\Processor[0], \Processor[1] and~\Processor[2] have processed their assigned
chunks in parallel, the results from the individual chunks need to be combined to
derive the overall result of the matching computation.
Combining proceeds from the first to the last chunk by propagating the resulting DFA state
from the previous chunk as the initial state for the following chunk. According to
Fig.~\ref{fig:motivDFA}, the DFA from our motivating example will be in state~\State[0] after matching
chunk~\Chunk[0]. State~\State[0] is propagated as the initial state for chunk~\Chunk[1].
Processor~\Processor[1] has matched chunk~\Chunk[1] for both possible initial states, i.e.,
\State[0] and \State[1], from which we obtain that state~\State[0] at the beginning
of chunk~\Chunk[1] takes the DFA to state~\State[1] at the end of chunk~\Chunk[1].
Likewise, the matching result for chunk~\Chunk[2] is now applied to derive state~\State[1]
as the final DFA state.
  
To compute the speedup over sequential DFA matching, we note that
processor~\Processor[0] processes 4~input characters, whereas processors~\Processor[1]
and~\Processor[2] match the assigned chunks
twice, for a total of 8~characters per processor. The resulting speedup is thus
$\frac{12}{8}$
or~$1.5$ (Combining the matching results will induce slight additional costs on the
order of the number of chunks, as we will consider in Section~\ref{sec:Algorithms}). 

\begin{figure}[H]
 \centering
 \begin{tabular}{lll}
    \rule{21mm}{0pt}&\rule{\colwidth}{0pt}&\rule{\colwidth}{0pt}\\[-2.2mm]
    \Chunk[0]&\Chunk[1]&\Chunk[2]\\
    \fbox{$a\,a\,a\,a\,a\,a$\rule{0pt}{\StrH}}&
    \fbox{$a\,b\,c$}&
    \fbox{$c\,c\,c$\rule{0pt}{\StrH}}\\[\iskip]
    \Processor[0]: \StartState&
    $\Processor[1]$: $\StartState, \State[1]$&
    $\Processor[2]$: $\StartState, \State[1]$
 \end{tabular}
 \caption{Balanced input partition according to the number of
states matched by a processor.\label{fig:balanced}}
\end{figure}

An input partition that accounts for the work imbalance between the initial and all subsequent
chunks is depicted in Fig.~\ref{fig:balanced}.
Because processors~\Processor[1] and~\Processor[2] match chunks for two states each,
their chunks are only half the size of the chunk assigned to processor~\Processor[0].
All processors now process 6~characters each, resulting in a balanced load and a $2$x speedup
over sequential matching.

By considering the structure of DFAs, the amount of redundant, speculative
computation can be reduced. For the DFA in Fig.~\ref{fig:motivDFA},
we observe that for each alphabet character $x\in\Sigma=\{a,b,c\}$,
there is only one DFA state (except the error state~\ErrorState) with an incoming transition labeled~$x$. 
Thus, this particular DFA has the structural property that for any given character~$x\in\Sigma$,
the DFA state after matching character~$x$ is known {\em a-priory\/} to be either
the error state or the state with the incoming transition labeled~$x$.

\def\bsize{12.5pt}
\def\coff{60pt}
\def\ccoff{120pt}
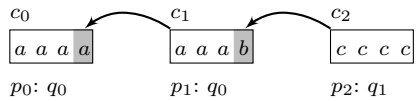
\begin{figure}[H]
     \centering\vspace{3mm}
     \begin{tikzpicture}
        %
        %
        \draw[fill=lightgray,lightgray](24pt,0pt) rectangle (31pt,\bsize);
        \draw[](0pt,0pt) rectangle (31pt,\bsize);
        \node at (4.0pt, 4.8pt) {$a$};
        \node at (12.0pt, 4.8pt) {$a$};
        \node at (20.0pt, 4.8pt) {$a$};
        \node at (28pt, 4.8pt) {$a$};
        \node at (4.0pt,18pt) {$\Chunk[0]$};
        \node at (-2.8pt,-10pt) [right] {$\Processor[0]$: $\StartState$};
        %
        %
        \draw[fill=lightgray,lightgray]($(\coff,0pt)+(24pt,0pt)$) rectangle ($(\coff,0pt)+(31pt,\bsize)$);
        \draw[]($(\coff,0pt)+(0pt,0pt)$) rectangle ($(\coff,0pt)+(31pt,\bsize)$);
        \node at ($(\coff,0pt)+(4.0pt, 4.8pt)$) {$a$};
        \node at ($(\coff,0pt)+(12pt, 4.8pt)$) {$a$};
        \node at ($(\coff,0pt)+(20pt, 4.8pt)$) {$a$};
        \node at ($(\coff,0pt)+(27.5pt, 5.8pt)$) {$b$};
        \node at ($(\coff,0pt)+(4.0pt,18pt)$) {$\Chunk[1]$};
        \node at ($(\coff,0pt)+(-2.8pt,-10pt)$) [right] {$\Processor[1]$: $\StartState$};
        %
        %
        \draw[]($(\ccoff,0pt)+(0pt,0pt)$) rectangle ($(\ccoff,0pt)+(31pt,\bsize)$);
        \node at ($(\ccoff,0pt)+(4.0pt, 4.8pt)$) {$c$};
        \node at ($(\ccoff,0pt)+(12pt, 4.8pt)$) {$c$};
        \node at ($(\ccoff,0pt)+(20pt, 4.8pt)$) {$c$};
        \node at ($(\ccoff,0pt)+(28pt, 4.8pt)$) {$c$};
        \node at ($(\ccoff,0pt)+(4.0pt,18pt)$) {$\Chunk[2]$};
        \node at ($(\ccoff,0pt)+(-2.8pt,-10pt)$) [right] {$\Processor[2]$: $\State[1]$};
        %
        %
        %
        \draw [arrows={-latex'},thick]($(\coff,0pt)+(0.0pt,13pt)$)
              parabola bend ($(\coff,0pt)+(-18.0pt,19pt)$)
              ($(\coff,0pt)+(-32.0pt,13pt)$);
        \draw [arrows={-latex'},thick]($(\ccoff,0pt)+(0.0pt,13pt)$)
              parabola bend ($(\ccoff,0pt)+(-18.0pt,19pt)$)
              ($(\ccoff,0pt)+(-32.0pt,13pt)$);
     \end{tikzpicture}
     \caption{Balanced input partition considering one symbol
              reverse lookahead.\label{fig:imax}}
\end{figure}

A processor can exploit this structural DFA property by performing
a {\em reverse lookahead\/} to determine the last character from the previous
chunk. From this character the DFA state at the beginning of the current chunk
can be derived. In Fig.~\ref{fig:imax}, the reverse lookahead for our motivating
example is shown. 
Reverse lookahead characters are shaded in gray.
Character~$a$ is the lookahead character in chunk~\Chunk[0]; only DFA state~\StartState\ 
from Fig.~\ref{fig:motivDFA}
has an incoming transition labeled~$a$, thus the DFA
must be in state~\StartState\ 
at the beginning of chunk~\Chunk[1].
Likewise, the DFA must be in state~\State[1] at the beginning of chunk~\Chunk[2],
because state~\State[1] is the only DFA state with an incoming transition
labeled~$b$ (the lookahead character of chunk~\Chunk[1]).
Note that for these considerations the error state~\ErrorState\ can be ignored,
because once a DFA has reached the error state, it will stay there (e.g.,
see Fig.~\ref{fig:motivDFA}). Thus, to compute the DFA matching result
it is unnecessary to process the remaining input characters
once the error state has been reached.

Because now all processors have to match only a single state per chunk, the chunks
are of equal size. For three processors, we achieve a speedup of $3$x over sequential matching for the
motivating example.

It should be noted that in the general case the structure of DFAs will be less ideal, i.e., there
will be more than one state with incoming transitions labeled by a particular input character.
Consequently, each chunk will have to be matched for more than one DFA state.
We will develop a measure for the suitability of a DFA for this type of speculative parallelization
in Section~\ref{sec:Algorithms}. Our analysis of the time-complexity of this method
shows that for $\NrProcessors>1$, a speedup is achievable in general. This has been
confirmed by our experimental evaluations on
SIMD, shared-memory multicore, and the Amazon EC2 cloud-computing
environments.
We will discuss the trade-offs that come with multi-character reverse lookahead, and we will
incorporate inhomogeneous compute capacities of processors to resolve load imbalances.
This is essential to effectively utilize heterogeneous multicore architectures, and to overcome
the performance variability of nodes reported with cloud computing environments~\cite{EC2Perf,Armbrust2009}. 

\section{Speculative DFA Matching}\label{sec:Algorithms}
Our speculative DFA matching approach is a general method, which
allows a variety of algorithms that differ with respect to the
underlying hardware platform and the incorporation of
structural
DFA properties. 
We start this section with the formalization of our basic speculative
DFA matching example from Section~\ref{sec:Overview}.
We then present our approach to exploit structural DFA properties to speed up
parallel, speculative DFA matching. Section~\ref{sec:Implementation} contains variants tailored
for SIMD, shared memory multicores and cloud computing environments.
 
\subsection{The Basic Speculative DFA Matching Algorithm\label{sec:basicAlgo}}
Our parallel DFA membership test consists of the following four steps; the first
step is only required on platforms with processors of inhomogeneous performance.
\begin{enumerate}
\item Offline profiling to determine the DFA matching capacity of each participating processor, 
\item partitioning the input string into chunks such that the 
      utilization of the parallel architecture is maximized,
\item performing the matching process on chunks in parallel such that
      redundant computations are minimized, and
\item merging partial results across chunks to derive the overall result of the matching
      computation.
\end{enumerate}

\begin{figure}[t]
    \centering
    \begin{tabular}{c}
    \begin{tikzpicture}[->,>=stealth']
        \tikzstyle{every initial by arrow}=[->]
        \node[smallstate, initial] (S0){\StartState};
        \node[smallstate,
              right of = S0,
              yshift= 3em,
              node distance = 4em] (S1){\State[1]};
        \node[smallstate,
              right of = S0,
              yshift= -3em,
              node distance = 4em] (S2){\State[2]};
        \node[smallfinalstate,
              right of = S0,
              node distance = 8em] (S3){\State[3]};
        \path   (S0)    edge[bend left]    node [above left]{$a$}  (S1)
        (S0)    edge[bend right]           node[above right]{$b$}      (S2)
        (S1)    edge[bend left]     node[above right]{$b$}      (S3)
        (S2)    edge                node[xshift=-1.5mm]{$a$}    (S1)
        (S2)    edge[bend right]    node[below right]{$b$}      (S3)
        (S3)    edge[loop right]    node{$a$}                 ()
        ;
    \end{tikzpicture}\makebox{\hspace{6mm}}\\
    {\small(a)}\\\\
    $\STR=bababbababbababbaaabbababbbaabbaaaba$\\
    {\small(b)}\\
    \end{tabular}
    \caption{Example DFA~(a) and input string with 36~symbols~(b). }
    \label{fig:ExampleDFA}
\end{figure}
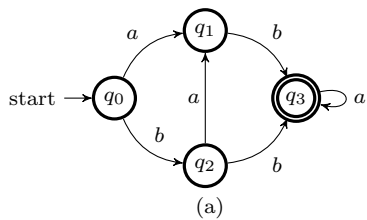

\begin{table}
    \centering
    \begin{tabular}{|c|c|c|c|c|}
    \hline
    Processor&$\Matched[k]$&$\Weight[k]$&\Length[0]$\cdot\Weight[k]$
    &Input character range\\
    \hline\hline
    \Processor[0]&$50$&$1.5$&$28.8$&$0$--$27$\\
    \hline
    \Processor[1]&$25$&$0.75$&$3.6$&$28$--$31$\\
    \hline
    \Processor[2]&$25$&$0.75$&$3.6$&$32$--$35$\\
    \hline
    \end{tabular}
    \caption{Computation of chunk sizes for Fig.~\ref{fig:ExampleDFA} and three processors of non-uniform
             processing capacities.}
    \label{tab:ExampleChunkComp}
\end{table}

{\bf Offline Profiling:}
For environments with inhomogeneous compute capacities, our 
offline profiling step determines the DFA matching
capacities of all participating processors. This information is required to partition
work equally among processors and thus balance the load.
With
heterogeneous multicore hardware architectures such as
the Cell~BE~\cite{CellAccelerators}, offline profiling must 
be conducted only once to determine
the performance of all types of processor cores provided by the architecture.
With cloud computing environments such as the Amazon EC2 cloud~\cite{EC2}, users only
have limited control on the allocation of cloud computing nodes. Moreover,
the performance of cloud computing nodes has been found
to differ significantly, which is by a large extent
attributed to variations in the employed hardware platforms~\cite{EC2Perf,Armbrust2009}. 
To compensate for the performance variations between cloud computing nodes,
offline profiling will be conducted at cluster startup time. Profiling cluster nodes
in parallel takes only on the order of milliseconds, which makes the overhead from
profiling negligible compared to the substantial cluster startup times
on EC2 (on the order of minutes
by our own experience and
also reported in~\cite{Ostermann09}).

To account for performance variations,
we introduce
a weight factor~\Weight[k], which denotes the processor capacity
of a processor~\Processor[k], normalized by the average
processor capacity of the system. 
On each processor~\Processor[k],
our profiler performs several partial sequential DFA
matching runs for a predetermined
number of input symbols on a given benchmark DFA. From the median of the obtained execution
times, we compute the number of symbols~$\Matched[k]$ matched by processor~\Processor[k] per microsecond.
The processor's
weight factor~$w_{k}$ is then
computed as
\begin{equation}
    \Weight[k]=\Matched[k]\cdot
    \left(\frac{1}{\NrProcessors}
    \cdot\sum_{0\le i <\NrProcessors}\Matched[i]\right)^{-1}
    \negthickspace\negthickspace\!.
    \label{eq:Weights}
\end{equation}
Columns~``$\Matched[k]$''
and ``$\Weight[k]$'' of Table~\ref{tab:ExampleChunkComp} contain example matching capacities
and corresponding weights for a system of three processors. 
We will apply processor weights to partition
the input string into chunks as follows.

{\bf Input Partitioning:}\label{sec:BasicMatching}
We observed already with our motivating example
from Fig.~\ref{fig:naive} that partitioning the input into
equal-sized chunks will result in load-imbalance: because
for the first chunk the initial DFA state is known to be~\StartState,
the first chunk needs to be matched only once.  
All other chunks must be matched for all possible initial states of the 
chunk , i.e., \NrStates~times,
in the worst case. In what follows, we will derive a partition of the input \STR\
into \NrProcessors~chunks, assuming that all except the first chunk need
to be matched for \NrStates\ states. In Section~\ref{sec:OptMatch}, we will exploit
structural DFA properties to reduce the number of states to be matched per chunk.

Intuitively, because processor~\Processor[0] has to match chunk~\Chunk[0] only once,
it can process a larger portion of the input~\STR\ than the processors assigned to subsequent chunks.
(This was observed already in Fig.~\ref{fig:balanced}, where chunk sizes were adjusted such that all
processors processed the same number of characters from the input.)
The objective of our optimization is to determine chunk sizes in such
a way that the processing times for all chunks are equal.
The purpose of the following equations is to compute a partition
of the input into chunks~\Chunk[i], $0\le i <\NrProcessors$,
where chunk~\Chunk[i] is a sequence of symbols from the input allocated 
to processor~\Processor[i].

Let \Length[i] denote the length of chunk~\Chunk[i] 
when $0\le i <\NrProcessors$, and $n$ be the
length of the input~\STR. Let us further assume that matching
of a character from the input takes constant time.
Processor~\Processor[0] matches chunk~\Chunk[0] from starting 
state~\StartState.
All other chunks need to be matched for all possible initial states.
To keep
work among processors balanced, chunk~\Chunk[0] must 
be \NrStates\ times longer than the other chunks, i.e., it must hold that
\begin{equation}
    \Length[i]=\frac{\Length[0]}{\NrStates},\text{ for } 
        1\le i <\NrProcessors.
    \label{eq:LI}
\end{equation}
The lengths of all chunks must add up to $n$, namely
\begin{equation}
    \sum_{0\le i<\NrProcessors}\Length[i]=n.
    \label{eq:NonWeighted}
\end{equation}
If processors have non-uniform processing capacity, we
incorporate weight factors from Eq.~\eqref{eq:Weights}, such that
weighted chunk sizes must add up to~$n$.
\begin{equation}
    \sum_{0\le i<\NrProcessors}\Length[i]\Weight[i]=n.
    \label{eq:WeightedL}
\end{equation}
Finally we solve the unknown~\Length[0] by substituting corresponding
parts of Eq.~\eqref{eq:LI} and~\eqref{eq:Weights} in Eq.~\eqref{eq:WeightedL},
i.e.,
\begin{equation}
    \Length[0]=\frac{n\cdot\NrStates}
                {\Weight[0]\cdot\NrStates
                  +\sum_{1\le i<\NrProcessors}\Weight[i]}.
    \label{eq:LZero}
\end{equation}
The start and end positions
for each chunk~\Chunk[k], $0\le k <\NrProcessors$, are computed by
the following equations. (Note that for $k=1$, the range for the sum over
$\Length[0]\Weight[i]$ is~$0$.) 

\begin{equation}
    \StartPos(\Chunk[k])=
    \begin{cases}
        0, & \text{for } k=0,\\[2.2mm]
        \left\lfloor\Length[0]\Weight[0] + \frac{1}{\NrStates}
        \sum_{1\le i<k} \Length[0]\Weight[i]\right\rfloor ~~~~~~&\text{otherwise}~~~~~~~~~~~\\
    \end{cases}
\label{eq:StartPosition}
\end{equation}

\vspace{1.2mm}

\begin{equation}
    \EndPos(\Chunk[k])=
    \begin{cases}
        n-1, &\text{for } k=|P|-1,\\[2.2mm]
        \lfloor\Length[0]\Weight[0] + \frac{1}{\NrStates}
        \sum_{1\le i\le k} \Length[0]\Weight[i]\rfloor -1, 
        \:\:\:&\text{otherwise}\\
    \end{cases}
\label{eq:EndPosition}
\end{equation}

An example DFA and an input string of length~$n=36$ are presented in
Fig.~\ref{fig:ExampleDFA}. The corresponding chunk sizes for three
processors with different processing capacities
are depicted in Table~\ref{tab:ExampleChunkComp}.
We observe by Eq.~\eqref{eq:LZero} that the length~\Length[0] of chunk~\Chunk[0] is $19.2$~characters,
and the weighted length according to processor weight~\Weight[0] is 28.8~characters.
From Eq.~\eqref{eq:LI} we observe that the remaining chunks are four times shorter than chunk~\Chunk[0],
because they have to be matched for $\NrStates=4$ states.
The weighted lengths of chunks~\Chunk[1] and~\Chunk[2] are thus $3.6$ characters each.
The rightmost column of Table~\ref{tab:ExampleChunkComp} depicts the character ranges of the input
as they have been assigned to each chunk.

{\bf Matching of Chunks:}
Algorithm~\ref{algo:basicmatch} depicts our basic speculative DFA
matching procedure.
We employ the notation introduced in~\cite{Holub:2009}
to denote a mapping of {\em possible
initial states\/} to {\em possible last active states\/} of a chunk.
This mapping is required to store a chunk's matching results for all possible
initial states. After matching chunks in parallel, the computed mappings will be used
to derive the overall DFA matching result.
Formally, this mapping is defined as a
vector~$$\StateMap[i]=[l_0, l_1,\ldots,l_{\NrStates-1}],$$
where~$0\le i <\NrProcessors$ and $l_j \in \States$ for all 
$0\le j <\NrStates$. Let element~$l_j$ of $\StateMap[i]$ denote the last 
active state, assuming that processor~\Processor[i] starts in 
state~\State[j] and processes the DFA membership test on 
chunk~\Chunk[i], i.e.,
$\TFE(\State[j],\Chunk[i])=l_j$. 

As an example, we consider chunk~\Chunk[2] from Fig.~\ref{fig:naive} and
the DFA from Fig.~\ref{fig:motivDFA}.
Chunk~\Chunk[2] will be matched for the possible initial states~\StartState\ and~\State[1], with
the resulting last active states~\ErrorState\ and~\State[1] and the result
vector~$\StateMap[2]=[\ErrorState,\State[1]]$. The meaning of vector~$\StateMap[2]$ is that
if the DFA assumes state~\StartState\ at the beginning of chunk~\Chunk[2], then it will
be in state~\ErrorState\ after matching chunk~\Chunk[2]. If the DFA assumes
 state~\State[1]\ at the beginning of chunk~\Chunk[2], then it will
be in state~\State[1]\ after matching chunk~\Chunk[2]. 

Our basic speculative DFA matching procedure employs Eqs.~\eqref{eq:StartPosition}
and~\eqref{eq:EndPosition}
to derive the start and end position of each chunk (lines 4--5
of Algorithm~\ref{algo:basicmatch}). The algorithm
distinguishes between the first input
chunk (lines~6--8) and all subsequent chunks (lines 9--12).
According to our partitioning scheme, chunk~\Chunk[0] is only matched
for the start state~\StartState\ (lines 7--8). For all subsequent chunks~\Chunk[i],
all possible DFA states are matched and stored in vector~$\StateMap[i]$. 
Chunk sizes are chosen according to processor weights and the number of states
to be matched with each chunk. The goal of this partitioning is to
load-balance the DFA matching to
effectively utilize the underlying parallel hardware platform.
We will discuss in Section~\ref{sec:TimeComplexity} that our partitioning
scheme makes this speculation failure--free.
The output of Algorithm~\ref{algo:basicmatch} is the set of 
vectors~$\StateMap[i]$,
where each vector describes the possible last states
according to the possible initial states of a given chunk.

\begin{algorithm}[htb]
    \SetKwInOut{Input}{Input}
    \SetKwInOut{Output}{Output}
    \SetKwData{START}{Start}
    \SetKwData{END}{End}
    \SetKwComment{EOLC}{//}{}

    \Input{\TF, \States,
           $\Sigma$,
           $\Processors$,
           $\STR=\Chunk[0]\Chunk[1]\ldots\Chunk[\NrProcessors-1]$
          }
    \Output{vector~$\StateMap[i]$ for each chunk~\Chunk[i]}
    \ForPar{$i\leftarrow0$ \KwTo $\NrProcessors-1$}{
        \For{$j\leftarrow0$ \KwTo $\NrStates-1$}{
            $\StateMap[i][j]\leftarrow j$
            \tcp*{initialize vector $\StateMap[i]$}
        }
        $\START\leftarrow\StartPos(\Chunk[i])$\\
        $\END\leftarrow\EndPos(\Chunk[i])$\\
        \eIf(\tcp*[f]{chunk~\Chunk[0]}){$i=0$}{
            \For{$k\leftarrow \START$ \KwTo $\END$}{
                $\StateMap[0][0]
                 \leftarrow\TF(\StateMap[0][0], \STR[k])$
            }
        }(\tcp*[f]{chunks~$\Chunk[1]\ldots\Chunk[{\NrProcessors-1}]$})
        {
            \ForEach{$j \in\States$}{
                \For{$k\leftarrow\START$ \KwTo $\END$}{
                    $\StateMap[i][j]
                    \leftarrow\TF(\StateMap[i][j], \STR[k])$
                }
            }
        }
    }
    \caption{Basic speculative DFA matching}
    \label{algo:basicmatch}
\end{algorithm}

{\bf Merging of Partial Results:}
After matching chunks in parallel, each
processor~\Processor[i]\ has constructed a mapping~\StateMap[i]\ of
possible initial states to last active states. To finish the DFA run,
the partial results computed for chunks~\Chunk[i]
need to be combined
to determine the last active state for the DFA-run over the whole
input string $\STR=\Chunk[0]\Chunk[1]\ldots\Chunk[\NrProcessors-1]$.
Chunk~\Chunk[0] is the only chunk for which we know the initial
state of the automaton, i.e., \StartState. We use this information
to apply the mappings~\StateMap[i] sequentially to derive the last
active state as follows~(it should be noted that index~$0$ of the $\StateMap[][\ldots]$ mapping
is the index of the start state~\StartState):
\begin{equation}
    \mbox{last active state}=
     \StateMap[\NrProcessors-1][
     \StateMap[\NrProcessors-2][
     \ldots\StateMap[0][0]\ldots]].
    \label{eq:SeqReduction}
\end{equation}
It has been shown in~\cite{Holub:2009} how a binary 
reduction~(see~\cite{Snyder}) can be used to parallelize this computation.
A binary reduction uses a combining operation
on two maps~$\StateMap[i]$
and~$\StateMap[j]$ to derive the combined
map~$\StateMap[i,j]$
as depicted in Eq.~\eqref{eq:MapReduction}.
\begin{equation}
\StateMap[i,j]=
\left[
\begin{array}{c}%
\StateMap[j][\StateMap[i][{0}]]\\
\StateMap[j][\StateMap[i][{1}]]\\
\vdots\\
\StateMap[j][\StateMap[i][{\NrStates-1}]]
\end{array}
\right]
\label{eq:MapReduction}
\end{equation}
The reduction step above can be performed repeatedly in parallel to
combine maps until we finally arrive at the
map~$\StateMap[0,\NrProcessors-1]$ which represents the overall effect of
a DFA. In particular, the value~$\StateMap[{0},{\NrProcessors-1}][0]$
will be the last active state of a DFA's run on the input~\STR.

The work in~\cite{Holub:2009} does not provide an evaluation of the 
relative merits
of sequential vs.~parallel merging of $\StateMap$-vectors. In particular,
the details of the employed parallel reduction algorithm are not specified.
We conducted experiments on a 40-core shared memory node of the Intel 
MTL using a binary tree for the parallel reduction to find
that the computation associated with the merging of $\StateMap$-vectors 
is not large enough to justify the overhead of a parallel reduction. Especially 
the overhead from the synchronization required between each of 
the $\BigO(\log_2(\NrProcessors))$ reduction steps is costly. 

Moreover, the overhead becomes significant if communication 
cost between nodes are introduced such as with cloud computers. We describe our
findings on the overheads of intra-node and inter-node communication
with the EC2 computing cloud in detail in
Section~\ref{sec:Implementation}. 
Section~\ref{sec:Implementation} introduces a new $\StateMap$-vector
merging technique
to cope with the overhead on cloud computers.

In short, we applied the sequential merging from 
Eq.~\eqref{eq:SeqReduction} with shared-memory multicore architectures
and a new hierarchical merging technique for cloud computing architectures, 
which will be explained in Section~\ref{sec:Implementation}.

\subsection{Optimizations Based on Structural DFA Properties\label{sec:OptMatch}}
The amount of work associated with a given chunk is determined by (1)~the length of the chunk, and
(2)~the number of DFA states for which the chunk needs to be matched.
In the following, we will distinguish between the initial chunk~\Chunk[0], and {\em subsequent\/}
chunks~\Chunk[i], $i>0$.
Before matching the initial chunk~\Chunk[0], the DFA will be in the starting state~\StartState, thus
chunk~\Chunk[0] only needs to be matched for~\StartState. Prior to the matching of subsequent chunks,
the DFA may assume any state in the general case,
thus subsequent chunks need to be matched~\NrStates\ times (see, e.g., the
motivating example in Fig.~\ref{fig:naive}). 
In this section we will exploit structural properties of DFAs to deduce a potentially smaller number~$\MaxNrIStates\le\NrStates$
of states which is the upper bound of initial states for all subsequent chunks.

The best case, i.e., $\MaxNrIStates=1$, has already been observed with our motivating example DFA from Fig.~\ref{fig:motivDFA}.
For each character ~$\sigma\in\Sigma$ of this DFA, it holds that there is only one state targeted by a transition labeled~$\sigma$.
Irrespective of the particular input character~$\sigma$, the DFA can only assume a single state after matching character~$\sigma$.
(As mentioned previously, for these considerations we may safely disregard the error state~\ErrorState, because
from the error state no other state is reachable; thus, a DFA that reached the error state will stay there.)
If there is only one possible DFA state after matching an input character, it follows that the DFA can only
be in one state after matching the last character prior to each subsequent chunk.
Thus the DFA can only be in one possible state at the beginning of each subsequent chunk,
and we have~$\MaxNrIStates=1$.

In the general case, values for \MaxNrIStates\ can range between~$1$
and~\NrStates. In the remainder of this section, we will investigate how to deduce this 
\MaxNrIStates\ value for a particular DFA, and how this information
can be incorporated with our speculative DFA matching algorithm. We will consider
real-world DFAs from PCRE and PROSITE to find that for all considered DFAs it holds
that~$\MaxNrIStates<\NrStates$, and that this property can be used to improve DFA
matching performance.
We have already observed with the input partition in Fig.~\ref{fig:imax} that reducing the number
of initial states of subsequent chunks enables us to increase the sizes of subsequent chunks. 
Larger subsequent chunks will reduce the size of the initial chunk~\Chunk[0] in turn.
Because we adjust chunk sizes such that all chunks will be processed in the same amount of time, reducing
the size of the initial chunk~\Chunk[0] will reduce the overall execution time of the matching process.
The overarching reason for this performance improvement is that the reduction of potential initial
states reduces the total number of symbols that have to be matched per chunk.

This can be formalized as follows.
Let~\MaxNrIStates\ denote the maximum number of
possible initial states that the DFA may assume at the start over
all subsequent chunks. This maximum can be
different for each chunk, depending on the last character of the preceding chunk.
We assume that for \MaxNrIStates\ we pick the 
maximum value out of all possible sets of initial states 
over all chunks.
If $\MaxNrIStates<\NrStates$, then the length~\Length[0]
of chunk~\Chunk[0] reduces by Eq.~\eqref{eq:LZero}:

\begin{equation}
\begin{split}
\Length[0]&=\frac{n\cdot\MaxNrIStates}
                {\Weight[0]\cdot\MaxNrIStates
                  +\sum_{1\le i<\NrProcessors}\Weight[i]}
           <
            \frac{n\cdot\NrStates}
                {\Weight[0]\cdot\NrStates
                  +\sum_{1\le i<\NrProcessors}\Weight[i]}.
\end{split}
\label{eq:SpeedupIStates}
\end{equation}

To deduce the maximum value for \MaxNrIStates, we  
eliminate states that can never be the initial state for a given
chunk. For each character~$\Sym\in\Sigma$, a DFA will contain a number of states that
have an incoming transition labeled~$\Sym$. Thus, if the last character of
a chunk's preceding chunk is $\Sym$, then only the states with an incoming transition
labeled~$\Sym$ need to be matched. We call the last input character of a chunk's preceding
chunk the {\em reverse lookahead symbol}. The number of states to be matched for
a reverse lookahead symbol~$\Sym\in\Sigma$ is a static property of a DFA.
It will range between~$1$ and~\NrStates.
The maximum number of states to be matched over any reverse lookahead symbol constitutes 
an upper bound on \MaxNrIStates, i.e., an upper bound on the number of states to be matched for any
subsequent chunk. Because \MaxNrIStates\ is a static DFA property, we can use it
to partition the input into chunks according to Eq.~\eqref{eq:SpeedupIStates}. At run-time,
a processor will use the reverse lookahead symbol
to determine the initial states to be matched for its assigned chunk.

Given lookahead symbol~$\Sym$, we define the set of initial states~$\IStates_\Sym$
as the set of all states that have an incoming transition labeled~$\Sym$.

\begin{equation}
    \IStates_{\negthinspace\Sym}=\{s: \TF(x, \Sym)=s\},
    \:\:\:
    \forall s,x\in\States.
    \label{eq:istates}
\end{equation}

If symbol~$\Sym$ is the reverse lookahead symbol of chunk~\Chunk[i],
then the set of possible initial states for chunk~\Chunk[i]
is~$\IStates_\sigma$.
We compute the set of possible initial states for all symbols from
the DFA's alphabet~$\Sigma$ and set~$\MaxNrIStates$
to the maximum cardinality among those sets, i.e.,
\begin{equation} 
    \MaxNrIStates=\max_{\Sym\in\Sigma}\left(\lvert\IStates_\Sym\rvert\right).
    \label{eq:MaxIstates}
\end{equation}

As an example, consider the DFA from 
Fig.~\ref{fig:ExampleDFA}(a).
Fig.~\ref{fig:InputString} shows the input string partitioned
for three processors of equal capacity, i.e.,
$\Weight[0]=\Weight[1]=\Weight[2]=1$.
The reverse lookahead symbols are depicted in gray.
No reverse lookahead is required for chunk~\Chunk[0], which will
be matched from the DFA's start state~\StartState.
Because the reverse lookahead symbol~$\sigma$
of chunk~\Chunk[1] is an `$a$', upon
matching of chunk~\Chunk[1] the DFA can only be in a state that
has an incoming transition labeled~`$a$'. Likewise, because 
the reverse lookahead symbol
of chunk~\Chunk[2] is `$b$', the DFA can only be in a state that
has an incoming transition labeled~`$b$' upon
matching of chunk~\Chunk[2].
We get $\IStates[a]=\{\State[1],\State[3]\}$, 
$\IStates[b]=\{\State[2],\State[3]\}$, and~$\MaxNrIStates=2$.
Inserting $n=36$, $\MaxNrIStates=2$, $\NrStates=4$ and
$\Weight[0]=\Weight[1]=\Weight[2]=1$
in Eq.~\eqref{eq:SpeedupIStates} yields $\Length[0]=18<24$ and
a speedup of $\frac{24}{18}=1.\dot{3}$ over the non-optimized
matching procedure.

\def\PSone{15.0pt}
\def\PStwo{14.0pt}
\def\Shift{1.08pt}
\def\d{-\Shift}
\def\u{+\Shift}
\def\n{+0pt}
\def\bsize{12.5pt}
\def\coff{60pt}
\def\ccoff{120pt}
\begin{figure}[htb]
\centering
\normalsize
     \begin{tikzpicture}
        \draw[fill=lightgray,lightgray] (56pt,0pt) rectangle (61.5pt,\bsize);
        \draw[](-26pt,0pt) rectangle (61.5pt,\bsize);
        \draw[fill=lightgray,lightgray] (120pt,0pt) rectangle (125.5pt,\bsize);
        \draw[](79.5pt,0pt) rectangle (125.5pt,\bsize);
        \draw[](143.5pt,0pt) rectangle (190.1pt,\bsize);
        \node (C1) at (4.0pt,  5.8pt) {$\STR:\:\:\:bababbababbababbaa$};
        \node[right = \PSone of C1.east,yshift=\n] (C2) {$abbababbb$};
        \node[right = \PStwo of C2.east,yshift=\n] (C3) {$aabbaaaba$};
        \node[above = 0mm of C1.north west,xshift=12mm, yshift=-.51mm] {$\Chunk[0]$};
        \node[above = 0mm of C2.north west,xshift=2.4mm, yshift=-.51mm] {$\Chunk[1]$};
        \node[above = 0mm of C3.north west,xshift=2.4mm, yshift=-.51mm] {$\Chunk[2]$};
        \node[below = 0mm of C1.south west,xshift=15.1mm, yshift=-.71mm] {$\Processor[0]$: $\StartState$};
        \node[below = 0mm of C2,xshift=-1.62mm, yshift=-.71mm] {$\Processor[1]$: $\State[1],\State[3]$};
        \node[below = 0mm of C3,xshift=-1.62mm, yshift=-.71mm] {$\Processor[2]$: $\State[2],\State[3]$};
        \draw [arrows={-latex'},thick](79.5pt,\bsize)
              parabola bend (68pt,\bsize+5pt)
              (58.5pt,\bsize);
        \draw [arrows={-latex'},thick](143.5pt,\bsize)
              parabola bend (132pt,\bsize+5pt)
              (122.5pt,\bsize);
     \end{tikzpicture}\vspace{-2mm}%
    \caption{Partitioned input string with reverse lookahead symbols and set of initial states to be matched for each chunk.
    \label{fig:InputString}}
\end{figure}

\begin{algorithm}[b]
    \SetKwInOut{Input}{Input}
    \SetKwInOut{Output}{Output}
    \SetKwData{START}{Start}
    \SetKwData{END}{End}
    \SetKwComment{EOLC}{//}{}
    \Input{\TF, \States,
           $\Sigma$,
           $\Processors$,
           $\STR=\Chunk[0]\Chunk[1]\ldots\Chunk[\NrProcessors-1]$,
           \StartState}
     \Output{vector~$\StateMap_{\Processor[i]}$ 
         for each chunk~\Chunk[i]}
     \ForEach{$\Sym_i\in\Sigma$}{
        $\IStates_{\Sym_i}\leftarrow\emptyset$\\
        \ForEach{$s\in\States$}{
           $\State[\text{target}]\leftarrow\TF(s,\Sym_i)$\\
           \If{$\State[\text{target}]\ne\ErrorState$}{
              $\IStates_{\Sym_i}\leftarrow\IStates_{\Sym_i}\cup\State[\text{target}]$
           }
        }
    }
    $\MaxNrIStates\leftarrow
         \max(\IStates_{\Sym_0},\ldots,\IStates_{\Sym_{|\Sigma|-1}})$\\
    \ForPar{$i\leftarrow0$ \KwTo $\NrProcessors-1$}{
        \For{$j\leftarrow0$ \KwTo $\NrStates-1$}{
            $\StateMap[i][j]\leftarrow j$
            \tcp*[f]{initialize vector $\StateMap[i]$}
        }
        $\START\leftarrow\StartPos(\Chunk[i], \MaxNrIStates)$\\
        $\END\leftarrow\EndPos(\Chunk[i], \MaxNrIStates)$\\
        \eIf(\tcp*[f]{chunk~\Chunk[0]}){$i=0$}{
           \For{$k\leftarrow \START$ \KwTo $\END$}{
                $\StateMap[0][0]
                 \leftarrow\TF(\StateMap[0][0], \STR[k])$
           }
        }(\tcp*[f]{chunks~$\Chunk[1]\ldots\Chunk[{\NrProcessors-1}]$})
        {
            \ForEach{$j \in\IStates[{\Chunk[i]}]$}{
                \For{$k\leftarrow\START$ \KwTo $\END$}{
                    $\StateMap[i][j]
                    \leftarrow\TF(\StateMap[i][j], \STR[k])$
                }
            }
        }
    }
    \caption{DFA matching applying initial state sets}
    \label{algo:MatchingInitialStates}
\end{algorithm}

Algorithm~\ref{algo:MatchingInitialStates} applies initial state sets
with the DFA matching procedure. Lines~1--7 compute
initial state sets~$\IStates_\Sym$ from Eq.~\eqref{eq:istates} and
\MaxNrIStates\ from Eq.~\eqref{eq:MaxIstates}. Unlike 
Algorithm~\ref{algo:basicmatch}, the partitioning is now based on the
maximum number of possible initial states, \MaxNrIStates, instead
of~\NrStates. The~$\StartPos$ and~$\EndPos$ functions
that compute the start and end position of each chunk now
receive~\MaxNrIStates\ as the second argument (lines~11--12 in
Algorithm~\ref{algo:MatchingInitialStates}). We updated
Eqs.~\eqref{eq:StartPosition} and~\eqref{eq:EndPosition} to
include an additional parameter to pass~\MaxNrIStates.
In
Eqs.~\eqref{eq:LZero}--\eqref{eq:EndPosition}, instead of
\NrStates\ we then use the provided argument value to partition the
input string and to compute the start and end position of each chunk.

Because the maximum number of initial states~\MaxNrIStates\ is a static property of a DFA,
it can be computed off-line. The overhead to compute~\MaxNrIStates\ can thus be avoided with
DFAs that are matched multiple times. For example, with 
protein patterns maintained in databases, corresponding DFAs can be expected to be matched on several 
DNA sequences.
However, with all our experiments, we computed~\MaxNrIStates\ online for every matching run (as stated in
Algorithm~\ref{algo:MatchingInitialStates}), to account for the general case were
a DFA is matched only once.

Another possible optimization of Algorithm~\ref{algo:MatchingInitialStates}
concerns the distribution of cardinalities of initial state sets~$\IStates_\Sym$. 
If the maximum value~\MaxNrIStates\ is significantly larger than the average, then
it is desirable to divide the input at boundaries with reverse lookahead symbols
that have a small initial state set.
This would further decrease the number of possible
initial states of subsequent chunks.
However, searching the input for the occurrence of particular characters 
constitutes an effort similar to the matching process itself.
Moreover, relying on statistical properties of the input string (i.e., the occurrence
of particular characters in the input) may violate the failure-freedom of our speculation: if a reverse lookahead
symbol with a low set of initial states cannot be found, then additional states need to be matched, resulting
in a possible speed-down.
In contrast,
by considering \MaxNrIStates\ states,
our optimization always shows equal or better
performance than the non-optimized matching procedure 
that has to match all states in~\States.

\subsection{Multiple Reverse Lookahead Symbols}
As discussed in the previous section, a smaller \MaxNrIStates\ constant
will decrease the number of symbols to be matched per chunk, thereby increasing
DFA matching performance. We can potentially decrease
the number of possible initial states, if we employ additional reverse lookahead symbols
with each chunk.
Given a string of reverse lookahead symbols~$\Sym_1\ldots\Sym_k$, $k\ge 1$.
We number the reverse lookahead symbols in the order they are matched by the DFA,
which is the reverse order of the lookahead itself.
The set
of initial states~$\IStates_{\Sym_1\ldots\Sym_k}$ constitutes the set of
all states that are the target of a path through the DFA labeled
by a string with postfix~$\Sym_1\ldots\Sym_k$, i.e.,
\begin{equation}
    \IStates_{\negthinspace\Sym_1\ldots\Sym_k}=\{s: \TFE(x, \Sym_1\ldots\Sym_k)=s\},
    \:\:\:
    \forall s,x\in\States.
    \label{eq:istatesm}
\end{equation}
Let~$\MaxNrIStates_{,r}$ be the maximum number of possible initial states
when using $r$~reverse lookahead symbols (in particular, $\MaxNrIStates_{,1}=\MaxNrIStates$).
Algorithm~\ref{algo:TwoSymISet} shows for a reverse lookahead of two characters  how
to compute initial state
sets~$\IStates_{\negthinspace\Sym_1,\Sym_2}$ and constant~$\MaxNrIStates_{,2}$.
The time complexity for computing~$\MaxNrIStates_{,r}$
is $\BigO\left(\lvert\Sigma\rvert^r\cdot\NrStates+\NrStates\right)$, i.e., the algorithm is exponential
in the number~$r$ of reverse lookahead symbols.
\begin{algorithm}[tb]
    \SetKwInOut{Input}{Input}
    \SetKwInOut{Output}{Output}
    \SetKwData{START}{Start}
    \SetKwData{END}{End}
    \SetKwComment{EOLC}{//}{}

    \Input{\TF, \States,
           $\Sigma$
          }
    \Output{$\IStates[\Sym_1\Sym_2],\MaxNrIStates_{,2}$}
    \ForEach{$\Sym_1\in\Sigma$}{
       \ForEach{$\Sym_2\in\Sigma$}{
          $\IStates[\Sym_1\Sym_2]\leftarrow\emptyset$\\
          \ForEach{$\State\in\States$}{
              $\IStates[\Sym_1\Sym_2]\leftarrow
               \IStates[\Sym_1\Sym_2]
               \cup
               \left(
                  \TF(\TF(\State,\Sym_1),\Sym_2)\setminus \{\ErrorState\}
               \right)$ 
          }
       }
    }
    $\MaxNrIStates_{,2}=\max_{\Sym_1,\Sym_2\in\Sigma}\left(\lvert\IStates[\Sym_1\Sym_2]\rvert\right)$
    \caption{Initial state set $\IStates_{\Sym_1\Sym_2}$ and $\MaxNrIStates_{,2}$
     computation for 2-character reverse lookahead}
    \label{algo:TwoSymISet}
\end{algorithm}

The following lemma establishes that
when increasing the amount of reverse lookahead symbols, the maximum number of possible initial
states~$\MaxNrIStates_{,r}$ of a DFA is bounded above by~\MaxNrIStates.

\begin{lemma}
Given a DFA, it holds that
$\MaxNrIStates=\MaxNrIStates_{,1}\ge\MaxNrIStates_{,2}\ge\ldots\ge\MaxNrIStates_{,\omega}$,
where $\omega$ denotes the length of the longest accepting path through the DFA.
\end{lemma}
\begin{proof}
Indirect. Without loss of generality
we assume a DFA with exactly one of its transitions labeled by a symbol~$\Sym\in\Sigma$, and state~\State\
being the target state of this transition. For this DFA, $\lvert\IStates[\Sym]\rvert=1$.
Given another symbol~$\Sym'\in\Sigma$, we assume that~$\lvert\IStates[\Sym'\Sym]\rvert=2$.
Then by the definition of~$\IStates_{\negthinspace\Sym'\Sym}$ in Eq.~\eqref{eq:istatesm}, 
this DFA must have two distinct states that are the target of a path labeled by a string with postfix~$\Sym'\Sym$.
However, this implies that these two target states have an incoming
transition labeled~$\Sym$, which contradicts our initial
assumption that~$\lvert\IStates[\Sym]\rvert=1$.
Thus for any two symbols~\Sym\ and~$\Sym'$, it holds that $\lvert\IStates[\Sym]\rvert\ge\lvert\IStates[\Sym'\Sym]\rvert$.
The extension to the general case
$\lvert\IStates[\Sym_1\ldots\Sym_k]\rvert\ge\lvert\IStates[\Sym'\Sym_1\ldots\Sym_{k}]\rvert$
is straightforward and the lemma follows.
\end{proof}

\subsection{Time Complexity}\label{sec:TimeComplexity}

The time complexity of sequential DFA matching is
$\BigO(n)$, where $n$ is the length of the input string.  Our basic
speculative DFA matching approach from Section~\ref{sec:basicAlgo}
distinguishes the first chunk from subsequent chunks
to partition the input string such that the matching load is balanced.
The time required for parallel matching is on the order of
\begin{equation}
    \BigO\left(\frac{n\cdot\NrStates}{\NrStates+\NrProcessors-1}\right).
    \label{eq:timecomp}
\end{equation}
The speedup of  Algorithm~\ref{algo:basicmatch} over sequential matching
is thus on the order of
\begin{equation}
   \BigO\left(1+\frac{\NrProcessors-1}{\NrStates}\right).
   \label{eq:speedup}
\end{equation}
It follows that in terms of algorithm complexity,
this approach
will not produce a speed-down, i.e., it
is failure-free. 

Eq.~\eqref{eq:pS_timecomp} shows the
time complexity of parallel DFA matching with reduced sets of
potential initial states from Section~\ref{sec:OptMatch}.  Because computing~$\MaxNrIStates_{,r}$
constitutes overhead,
we have an additional term
$\BigO(\NrStates\cdot\lvert\Sigma\rvert^r)$,
where $r$ is the number of reverse lookahead symbols.

\begin{equation}
    \BigO\left(\NrStates\cdot\lvert\Sigma\rvert^r
     +\frac{n\cdot\lvert\MaxNrIStates_{,r}\rvert}{\lvert\MaxNrIStates_{,r}\rvert+\NrProcessors-1}\right)
    \label{eq:pS_timecomp}
\end{equation}
If $n \gg \NrStates$, or if $\MaxNrIStates_{,r}$ is computed off-line,
the additional term can be neglected.
Even when computing $\MaxNrIStates_{,r}$ on-line, for all considered
cases the approach with reduced sets of
potential initial states showed better performance.

Our method is capable of utilizing processors of different processing
capacities, which is relevant for heterogeneous multiprocessors and
for cloud computing environments. Different processor
weights~\Weight[] encode processors' capacities. Because we
employ weights to calculate chunk sizes for processors, we
encode different processing capacities in the size of each processor's
chunk.  If we do not apply weights for processors of different
processing capacities, the following equation describes the overall time
complexity,

\begin{equation}
    \BigO\left(\frac{nm}{m+p-1}\right),
    \label{eq:w_timecomp}
\end{equation}
where $p=\NrProcessors\times\Weight[worst]$ and $\Weight[worst]= 
\min(\Weight[0],\Weight[1],...,\Weight[{|\Processors|-1}])$ and $m$ is
either $\vert\States\rvert$ or $\lvert\MaxNrIStates_{,r}\rvert$.
Incorporating reduced sets of potential initial
states with Eq.~\eqref{eq:speedup} yields a speedup on the order of
\begin{equation}
\BigO\left(1+\frac{\NrProcessors-1}{\NrStates\cdot\gamma}\right)\text{,}
\label{eq:obtained_speedup}
\end{equation}
where the ratio~$\gamma=\frac{\MaxNrIStates_{,r}}{\NrStates}$
of reduced sets of potential initial states~$\MaxNrIStates_{,r}$
to~\NrStates\ constitutes a structural DFA property.
We have~$0<\gamma\le 1$, where the magnitude of actual values
for~$\gamma$ is negatively correlated to the profitability of our
optimization for particular DFAs.

\section{Implementation}\label{sec:Implementation}
We implemented our speculative DFA matching algorithms for the three
architectures summarized in Table~\ref{tab:HWSpec}.
For our shared-memory multicore architecture implementation we were granted access 
to the Intel Academic Program Manycore Testing Lab~(Intel MTL, \cite{MTL}), which
is an experimental environment of non-commercial, 40-core nodes provided by Intel 
mainly for educational purposes. POSIX threads~\cite{Butenhof97} were used 
to parallelize DFA matching across multiple cores.
To vectorize our speculative DFA matching algorithm, we employed
version~2 of the Advanced Vector Extensions (AVX2) of the forthcoming
Intel Haswell CPU architecture~\cite{AVX2}.
The AVX2 instruction set provides
256~bit registers enabling 8-fold vectorization on 32-bit integer
and single precision floating point data types. 
AVX2 is the first x86 instruction set extension to provide a gather-operation
for vectorized indexed read operations from memory
(vectorized register-indirect addressing). To the best of our knowledge,
we are the first to utilize gather operations to vectorize DFA matching.
Because the Haswell architecture is scheduled to be released in 2013,
there is no processor available yet which supports  AVX2 instructions.
Hence, we used Intel's Software
Development Emulator~(SDE, \cite{SDE}) to emulate AVX2 instructions.
To evaluate our approach in a cloud computing environment,
we employed m2.4xlarge and cc2.8xlarge instances of
the Amazon EC2 elastic computing cloud~\cite{EC2}.
Each EC2 instance 
provides a nominal dedicated 
compute capacity stated in Amazon's proprietary Compute Unit~(CU) measure.
Hardware specifications of the used Amazon EC2 instance types (nodes) 
are given in Table~\ref{tab:HWSpec}. 
For our experiments, we employed 20~instances with a total of 320~physical cores.
For communication across threads, the MPI message passing interface was used.

\begin{table*}
    \centering
    \begin{tabular}{|m{1.8cm}|m{2.7cm}|c|c|c|c|}
    \hline
    Name & CPU Model & CPUs & $\frac{\mbox{Cores}}{\mbox{CPU}}$
    & Clock Freq. & Note\\
    \hline
    \hline
    Intel MTL & Intel Xeon E7-4860 & 4 & 10 & 2.27~GHz &n/a\\
    \hline
    SDE emulator on local server& AVX2/Haswell on Intel Xeon E5405 host & n/a & n/a& n/a &n/a\\
    \hline
    Amazon EC2 (m2.4xlarge) & Intel Xeon X5550 & 2 & 4 & 2.67~GHz & 26~EC2 CUs\\
    \hline
    Amazon EC2 (cc2.8xlarge) & Intel Xeon E5-2670 Sandy Bridge & 2 & 8 & 2.60~GHz & 88~EC2 CUs\\
    \hline
    \end{tabular}
    \caption{Hardware Specifications}
    \label{tab:HWSpec}
\end{table*}

\begin{figure}
\centering
\begin{tabular}{ccc}
    \resizebox{4cm}{!}{
    \begin{tikzpicture}[->,>=stealth']
        \tikzstyle{every initial by arrow}=[->]
        \node[smallstate, initial] (S0){\StartState};
        \node[smallstate,
              right of = S0,
              yshift= 3em,
              node distance = 4em] (S1){\State[1]};
        \node[smallstate,
              right of = S0,
              yshift= -3em,
              node distance = 4em] (S2){\State[2]};
        \node[smallfinalstate,
              right of = S0,
              node distance = 8em] (S3){\State[3]};
        \node[below of = S2, node distance = 8mm] (dummy){};
        \path   (S0)    edge[bend left]    node [above left]{$a$}  (S1)
        (S0)    edge[bend right]           node[above right]{$b$}      (S2)
        (S1)    edge[bend left]     node[above right]{$b$}      (S3)
        (S2)    edge                node[xshift=-1.5mm]{$a$}    (S1)
        (S2)    edge[bend right]    node[below right]{$b$}      (S3)
        (S3)    edge[loop right]    node{$a$}                 ()
        ;
    \end{tikzpicture}%
    }\rule{4mm}{0mm}&
    \begin{tabular}[b]{|ccc|} 
    \hline
    \multicolumn{2}{|c}{(START)$|$-}&0\\
    0 & a & 1\\
    0 & b & 2\\
    1 & b & 3\\
    2 & a & 1\\
    2 & b & 3\\
    3 & a & 3\\
    3 &\multicolumn{2}{c|}{-$|$(FINAL)}\\
    \hline
    \end{tabular}
    &
    \setlength\tabcolsep{1pt}
    \hskip4mm
    \begin{tabular}[b]{rlll}
    \texttt{SBase = \{}&\texttt{1,}&\texttt{2,}\\
    &\texttt{4,}&\texttt{3,}\\
    &\texttt{1,}&\texttt{3,}\\
    &\texttt{3,}&\texttt{4,}\\
    &\texttt{4,}&\texttt{4 \};}\\\rule{0mm}{6mm}\\
    \end{tabular}%
    \\[2mm]
    \small{(a)} &\small(b)&\small{(c)} \\[5mm]
    \multicolumn{3}{c}{\hskip15mm\STR=bababbababbababababa}\\[1mm]
    \multicolumn{3}{c}{\hskip15mm\texttt{IBase=\{1,0,1,0,1,1,0,1,0,1,1,0,1,0,1,0,1,0,1,0\};}}\\[2mm]
        \multicolumn{3}{c}{\hskip15mm(d)}
\end{tabular}
\caption{Example DFA~(a), Grail+ format~(b), \texttt{SBase} 1-dimensional
    transition table representation~(c), and representation of the
    DFA input~(d).}
\label{fig:SBaseExample}
\end{figure}
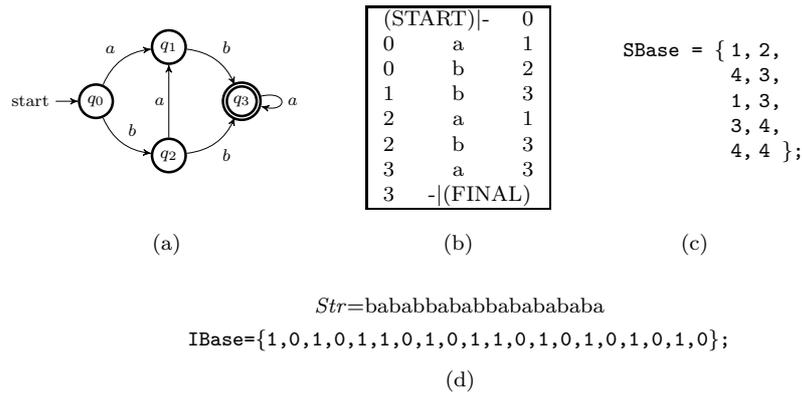

We tailored our DFA data-structures to maximize performance and to utilize
the AVX2 instruction set, in particular the novel AVX2 32-bit gather
operations.
To generate minimal DFAs from regular expressions, we use 
Grail+~\cite{Grail95,Grail}, which is a formal language toolset for the manipulation
and application of regular expressions and automata. 
Our DFA matching framework reads DFAs and input strings in Grail+ format and converts
them to our framework's internal representation.

DFA transition tables are usually represented as 2-dimensional
arrays, with rows  for each state and one column for each character~$x\in\Sigma$.
With our representation, 2-dimensional arrays are flattened into
consecutive, 1-dimensional arrays.
This representation allows to store multiple DFAs of different
alphabet sizes, and it facilitates application of AVX2 gather operations (i.e.,
gather operations allow 1-dimensional indexed reads only).
Fig.~\ref{fig:SBaseExample}(a) shows our running example DFA
from Fig.~\ref{fig:ExampleDFA} and the DFA's Grail+
format~(Fig.~\ref{fig:SBaseExample}(b)). Our
transition table representation is given in C-like pseudo
code in Fig.~\ref{fig:SBaseExample}(c).
Grail+ encodes DFA states as integers. Lines in Grail+ format
represent triples $\langle\negthinspace$~source-state, transition-label, target-state~$\negthinspace\rangle$, with
the start and accepting states indicated on separate lines. Our DFA representation
encodes states as row-indexes into the DFA transition table.
Note that State~4
represents the error-state~\ErrorState.
Row-indexes are calculated relative to the base address \texttt{SBase} of the array.
In case of a second DFA stored after the running example, the second DFA's row
indexes will also be stored relative to \texttt{SBase}.
For the input string, we introduce a 1-dimensional array \texttt{IBase} of integers. 
For example, in Fig.~\ref{fig:SBaseExample}(d), character~$a$ is mapped
to the value~$0$, and character~$b$ is mapped to~$1$.
Multiple DFA input strings may be concatenated in array~\texttt{IBase}.
Generation of this DFA and input string representation can be trivially implemented
while parsing the Grail+ DFA input data.
Our representation allows to run a single DFA simultaneously on multiple input strings, or to match
multiple DFAs on one or more input strings.

\begin{lstlisting}[language=C, 
    caption=Baseline matching routine in C for a possible initial state 
    of a chunk, label=list:Cmatching, ]
// Get address of first and last character of chunk:
INPUT_T* curPtr=&IBase[StartPos];
INPUT_T* endPtr=&IBase[EndPos];

// Get starting state and perform matching:
STATE_T CurrentState=InitialState*NrSymbols;
for ( ; curPtr!=endPtr; curPtr++) {
  CurrentState=SBase[CurrentState + *curPtr];
}
\end{lstlisting}

Listing~\ref{list:Cmatching} shows how a chunk is matched for one possible initial
state on multicore architectures. It should be noted
that by encoding the transition table's DFA states as offsets relative to the \texttt{SBase} base address,
2-dimensional table lookups of conventional DFA representations
are simplified to a 1-dimensional
lookup that avoids the rows-times-column multiplication of 2-dimensional
arrays---with our representation, we only
add the current state's offset to the current
input symbol (line~8 of Listing~\ref{list:Cmatching}).
We employ pointers to access the input and to detect loop termination, thereby avoiding the need
for maintaining a separate loop counter variable.
When compiled to x86-64, this matching loop consists of only two add operations, one
comparison, one indexed load and one conditional jump, which compares favorable to
Grail+'s matching loop implemented in C++, which requires more than an order of magnitude
more instructions for the same purpose. We used a variant of Listing~\ref{list:Cmatching} for sequential
DFA matching. This sequential matching routine was used as an efficient yardstick for the comparison to
our parallelized
matching algorithms, because we found our sequential matching routine to be more efficient than Grail+, and the closest approach
from the related work, i.e., \cite{Holub:2009}, incurred slowdowns over sequential matching (see
Fig.~\ref{fig:speedup_mtl_holub} and Section~\ref{sec:RelatedWork}).

\subsection{Vectorized DFA matching using AVX2 instruction set extensions}

\begin{lstlisting}[caption=Vectorized DFA matching of chunks using AVX2 intrinsics, label=list:AVX2matching]
int i;
__m256i InpSyms, Ones = _mm256_set1_epi32 (1);

// Load initial indices into SBase and IBase arrays:
__m256i States = _mm256_load_si256((__m256i const *) CStatesInit);
__m256i InpIdx = _mm256_load_si256((__m256i const *) CInputInit);

for (i = ChunkLength; i>0; i--) {
   // Load input characters from IBase, indexed by InpIdx:
   InpSyms = _mm256_i32gather_epi32 (IBase, InpIdx, 4);
   // Calculate indices of next states:
   States = _mm256_add_epi32 (States, InpSyms);
   // Load next state values from SBase, indexed by States:
   States = _mm256_i32gather_epi32 (SBase, States, 4);
   // increase input indices by one:
   InpIdx = _mm256_add_epi32(InpIdx, Ones);
}
\end{lstlisting}

Listing~\ref{list:AVX2matching} shows our core matching loop with 8-fold
vectorization employing AVX2 vector instruction intrinsics~\cite{AVX2}.
Data type~\texttt{\_\_m256i} represents an 8-way vector containing
8 32-bit \texttt{int} variables. Variables~\texttt{States} and
\texttt{InpIdx} contain the indices into the state transition
table~\texttt{SBase} and the input array~\texttt{IBase}. They
are initialized to precomputed starting-positions of chunks in lines~5 and~6. We use the
\texttt{\_mm256\_i32gather\_epi32} intrinsic to perform
vectorized, indexed loads from the \texttt{SBase} and \texttt{IBase} arrays.
For example, in line~8, 8~input characters are loaded from \texttt{IBase}. Note
that the offsets in vector \texttt{InpIdx} are scaled by a factor of~4 (the third argument of the intrinsic),
to account for the 32-bit size of type~\texttt{int}. 
For further details on the used intrinsics, we refer to~\cite{AVX2}.
The reason to count the loop
index variable down instead of up is because the decrement instruction
will already set the x86 CPU's sign flag when we cross zero. This way we 
save a \texttt{cmp} instruction which yields additional 12\% of performance 
improvement. Neither GCC nor Intel's ICC managed to generate optimal assembly code
from Listing~\ref{list:AVX2matching}, which required us to use inline assembly
instead. Auto-vectorization of sequential DFA matching is out of reach for compilers, because of the 
dependencies between current and next DFA state.

\subsection{DFA Matching on Cloud Computing Architectures\label{sec:cloudImplementation}}

\newcommand{\UNode}{\protect\begin{tikzpicture}[darkstyle/.style={circle,fill=black,minimum size=1.3mm, inner sep=0mm},
                      scale=1,
                      ->,>=stealth']
                      \node [circle,draw,minimum size=1.3mm,inner sep=0mm]  (50) at (0,0) {};
                      \end{tikzpicture}}
\def\NSkip{0.5}
\def\SSkip{4.5}
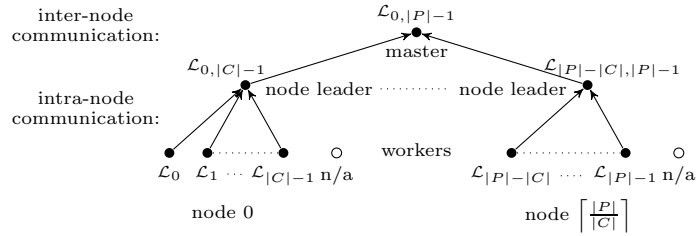
\begin{figure}[t]
    \centering
    \hspace{-6mm}%
    \begin{tikzpicture}[darkstyle/.style={circle,fill=black,minimum size=1.3mm, inner sep=0mm},
                        scale=1,
                        ->,>=stealth']
        \node [darkstyle]  (L1) at (\NSkip*3,2) {}; 
        \node [right, xshift=0.8mm, yshift=-0.5mm] (Leader1) at (L1.east) {\scriptsize node leader};
        \node [above,xshift=-4mm,yshift=-1.5mm] at (Leader1.north west) {$\scriptstyle \StateMap[0,\NrCores-1]$};
        \node [darkstyle]  (Master) at (0.5*\SSkip + \NSkip*3,2.7) {}; 
        \foreach \x in {1,2,4} {
             \node [darkstyle]  (\x0) at (\NSkip*\x,1.1) {}; 
             \draw [->] (\x0) -- (L1);
        }
        \node [circle,draw,minimum size=1.3mm,inner sep=0mm](50) at (\NSkip*5+0.2,1.1) {}; 
        \node [below]  (L10) at (10.south) {$\scriptstyle\StateMap[0]$}; 
        \node [below]  (L20) at (20.south) {$\scriptstyle\StateMap[1]$}; 
        \node [below]  (L40) at (40.south) {$\scriptstyle\StateMap[\NrCores-1]$}; 
        \node [below]  (L50) at (50.south) {\scriptsize n/a}; 
        \draw [-,style=dotted] (20) -- (40);
        \draw [-,style=dotted] (L20) -- (L40);
        \node []  (T1) at (\NSkip-1.1,1.7) {\begin{minipage}{3cm}\scriptsize\centering intra-node\\[-0.2mm]communication:\end{minipage}}; 
        \node []  (T2) at (\NSkip-1.1,2.8) {\begin{minipage}{3cm}\scriptsize\centering inter-node\\[-0.2mm]communication:\end{minipage}}; 

        \node [darkstyle]  (L2) at (\SSkip+\NSkip*3,2) {};
        \node [left, xshift=-1mm, yshift=-0.5mm] (Leader2) at (L2.west) {\scriptsize node leader};
        \node [above,xshift=5.0mm,yshift=-1.5mm] at (Leader2.north east) {$\scriptstyle \StateMap[\NrProcessors-\NrCores,\NrProcessors-1]$};
        \foreach \x in {1,4} {
             \node [darkstyle]  (N\x0) at (\SSkip+\NSkip*\x,1.1) {}; 
             \draw [->] (N\x0) -- (L2);
        }
        \node [circle,draw,minimum size=1.3mm,inner sep=0mm]  (P2) at (\SSkip+\NSkip*5+0.2,1.1) {}; 
        \node [below]  (LN10) at (N10.south) {$\scriptstyle\StateMap[\NrProcessors-\NrCores]$}; 
        \node [below]  (LN40) at (N40.south) {$\scriptstyle\StateMap[\NrProcessors-1]$}; 
        \node [below]  (P2) at (P2.south) {\scriptsize n/a}; 
        \draw [-,style=dotted] (N10) -- (N40);
        \draw [-,style=dotted] (LN10) -- (LN40);

        \draw [->] (L1) -- (Master);
        \draw [->] (L2) -- (Master);
        \node [below] at (Master.south) {\scriptsize master};
        \node [above,yshift=-0.6mm] at (Master.north) {$\scriptstyle \StateMap[0,\NrProcessors-1]$};
        \node (T4) at (0.5*\SSkip + \NSkip*3,1.14) {\scriptsize workers};
        \node [xshift=-3mm] (T5) at (\NSkip*3,0.27) {\raisebox{1mm}{\scriptsize node 0}};
        \node [xshift=-1.2mm] (T6) at (\SSkip+\NSkip*3,0.27) {\raisebox{1mm}
              {\scriptsize node~$\left\lceil\frac{\NrProcessors}{\NrCores}\right\rceil$}};
        \draw [-,style=dotted] (Leader1) -- (Leader2);
    \end{tikzpicture}\vspace{-2.8mm}%
    \caption[Hierarchical merging of \StateMap-vectors]{Hierarchical merging of \StateMap-vectors to reduce message delay and variability on EC2. The
             number of available processing cores is denoted by~\NrProcessors, and the number of
             cores allocated per node is denoted by~\NrCores. One core per node is left unallocated,
             to avoid performance degradation with hypervised EC2 nodes.
             Unallocated cores are denoted by
             symbol~``\UNode''.}
    \label{fig:Merge}
\end{figure} 

With our implementation for the EC2 cloud computing environment, 
we employed
MPI-based message-passing to communicate between cores
for merging $\StateMap$-vectors.  
The chosen MPI implementation was MPICH2 version~1.4~\cite{MPICH2}.
As mentioned previously, parallel reduction based on 
binary trees did not achieve satisfactory performance. 
We found the message transfer times of messages between EC2 nodes
too high to make binary reduction profitable. E.g., the average inter-node
transfer time for a single $\StateMap$-vector
was \SI{362}{\micro\second}, with a standard-deviation
of 3.6\%. In comparison, the same intra-node message would take on average only
\SI{2.68}{\micro\second},
with a standard-deviation of 0.14\%. 
This observation is in line with a recent study that reports large delay variations
and unstable network throughput for the EC2 cloud~\cite{Wang:2010}. 

To account for the message delay and variations on EC2, we devised 
a variant of parallel reduction that is hierarchical with respect to intra-node and inter-node communication.
This 2-tier merging approach is based on the observation that intra-node messages showed
substantially lower message transfer times and variations than inter-node communication.
Our reduction
proceeds in two steps, as depicted in Fig.~\ref{fig:Merge}.
In the first step, $\StateMap$-vectors are merged locally by a designated node leader.
In the second step, node leaders send their $\StateMap$-vectors 
to the master process which combines them to compute the overall matching result.
Without loss of generality, this 2-step merging scheme requires that on each EC2 node,
DFA-matching worker processes
are allocated to adjacent chunks.
Our worker-to-node allocation scheme is parameterized by the number of cores to utilize per node,
denoted by~\NrCores. For reasons explained below, we leave one core unallocated per EC2 node. 
Fig.~\ref{fig:Merge} depicts the computation of $\StateMap$-vectors by workers (for one chunk),
node leaders (the combined map over all chunks matched on a node) and the master (the overall map from
the first to the last chunk). 
Unallocated cores are denoted by
symbol~``\begin{tikzpicture}[darkstyle/.style={circle,fill=black,minimum size=1.3mm, inner sep=0mm},
                        scale=1,
                        ->,>=stealth']
        \node [circle,draw,minimum size=1.3mm,inner sep=0mm]  (50) at (0,0) {};
\end{tikzpicture}'').

Our two-tier merging scheme outperformed parallel binary reduction and sequential merging for even
the largest EC2 clusters (i.e.,  up to 20 nodes, 
which is the maximum possible EC2 cluster size~\cite{EC2limit}). 
We found  MPI messages among
processes on the same node to show both low latency and low variability.
We conjecture
that MPICH2 applies shared-memory message passing optimizations
similar to~\cite{Karonis:2000} for node-local communication. Moreover, node-local communication
is free from delay variations induced by the network that connects nodes.
Therefore, with our merging scheme the only communication step subjected to EC2's message variability is the
merging step conducted by the master.  
This compares favorably to any parallel reduction scheme with more than one reduction step involving
inter-node communication, because each such reduction step may suffer from message delays caused by the
underlying network. 

As mentioned above, we deliberately left one core per EC2 node
unallocated.
We observed that without sacrificing one core per EC2 node, there was a high probability that one of the workers on each node
would experience a matching performance on the order of one magnitude lower than the workers on the remaining cores.
This performance degradation did not affect the offline profiling step, for which we took the median of
a series of partial matching runs. However,
this performance degradation randomly showed
with DFA matching. Because we could not reproduce this problem
on a local cluster of Linux computers, we attribute this
performance degradation to EC2 hypervisor activities that occasionally preempted the execution of one arbitrary worker thread per node.
Leaving one core unallocated on EC2 eliminated this problem. Given the increasing numbers of cores per CPU, leaving one core
unallocated can be considered an increasingly small sacrifice (e.g., our experiments were conducted with EC2 nodes providing 8 and 16 cores,
respectively). 

\section{Experimental Results}\label{sec:ExperimentalResults}
We conducted experiments for both our basic and optimized speculative matching
algorithms 
and the approach of Holub and \v{S}tekr presented in~\cite{Holub:2009}.
We employed \NrPCREs~regular expressions from the PCRE library~\cite{PCRELib} 
and \NrPROSITEs~protein patterns from  the
PROSITE protein database~\cite{PROSITE}. Protein patterns were selected
as an example for the application domain of DNA 
sequence analysis.  We compared our algorithms to the baseline
sequential DFA matching algorithm from Section~\ref{sec:Implementation}
and to the currently used matching engine that comes with PROSITE. 
All 
PCRE regular expressions and PROSITE protein patterns were
converted to unique minimum DFAs using Grail+~\cite{Grail95,Grail}. 
All experiments except the experiments on EC2
were conducted with inputs of one million characters. 
Because we employed up to 288~cores on EC2, the problem sizes
of one million characters turned out too small for precise performance measurements.
Thus we used
inputs of 
8~million characters on EC2. We show the scalability of our approach on
both the MTL and the EC2 platforms for up to 10~billion characters in
Section~\ref{subsec:struct}.
Note that for increased readability we represent speed-downs by
negative values instead of fractional values. For example, conventional
denotation for a 2x~speed-down is $\frac{1}{2}$ but we use -2. 

\label{subsec:EvalPar}
\begin{figure*}[ht]
\centering
    \begin{tabular}{@{}c@{}c@{}}
        \subfigure[PROSITE speedups, incl.~\MaxNrIStates\ opt.]
            {\label{fig:speedup_mtl_par_opt_prosite}
             \hspace{-2.2mm}
             \includeGraphics[clip=true, height=4.3cm, trim=4mm 6mm 0 0]
                             {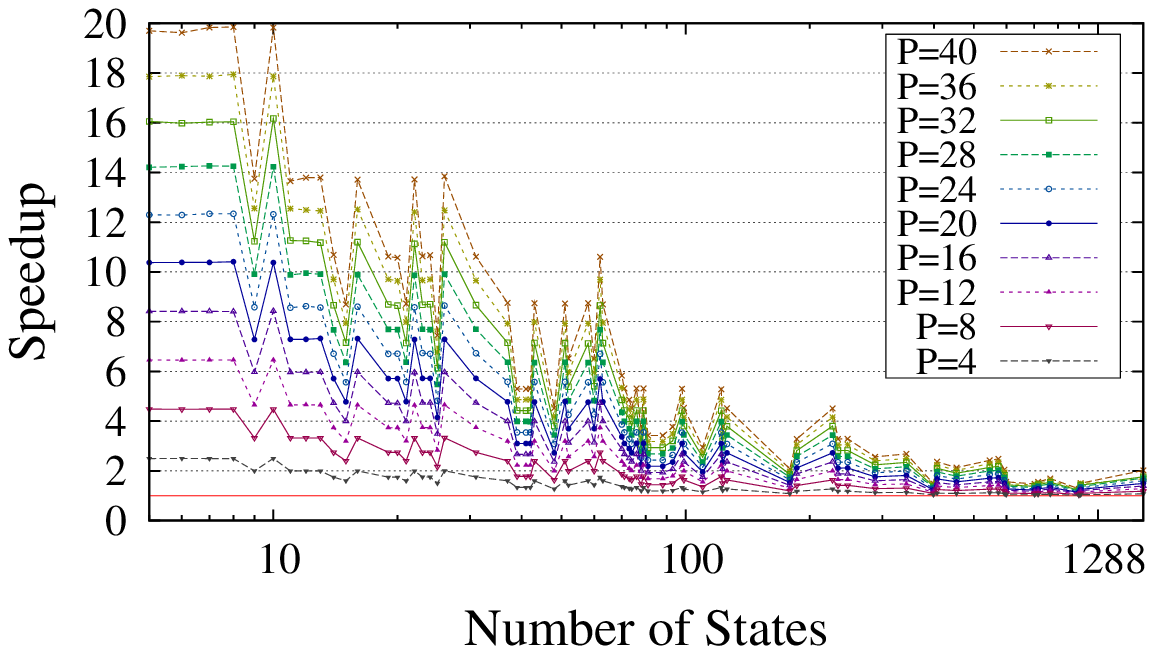}}
        &\hspace{-8mm}{
        \subfigure[PROSITE speedups of \MaxNrIStates\ opt.]
            {\label{fig:speedup_mtl_opt_against_noOpt_prosite}
             \includeGraphics[clip=true, height=4.3cm, trim=12mm 6mm 0 0]
                             {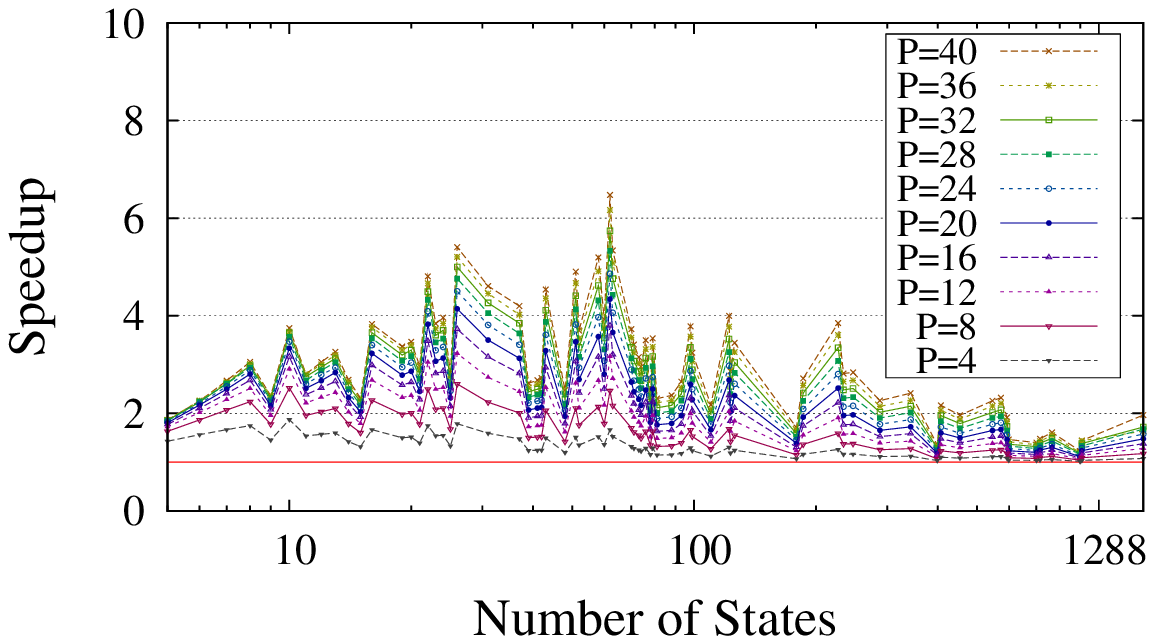}}}
        \vspace{-4mm}\\
        \subfigure[PCRE speedups, incl.~\MaxNrIStates\ opt.]
            {\label{fig:speedup_mtl_par_opt_re}
             \hspace{-2.2mm}
             \includeGraphics[clip=true, height=4.3cm, trim=4mm 6mm 0 0]
                             {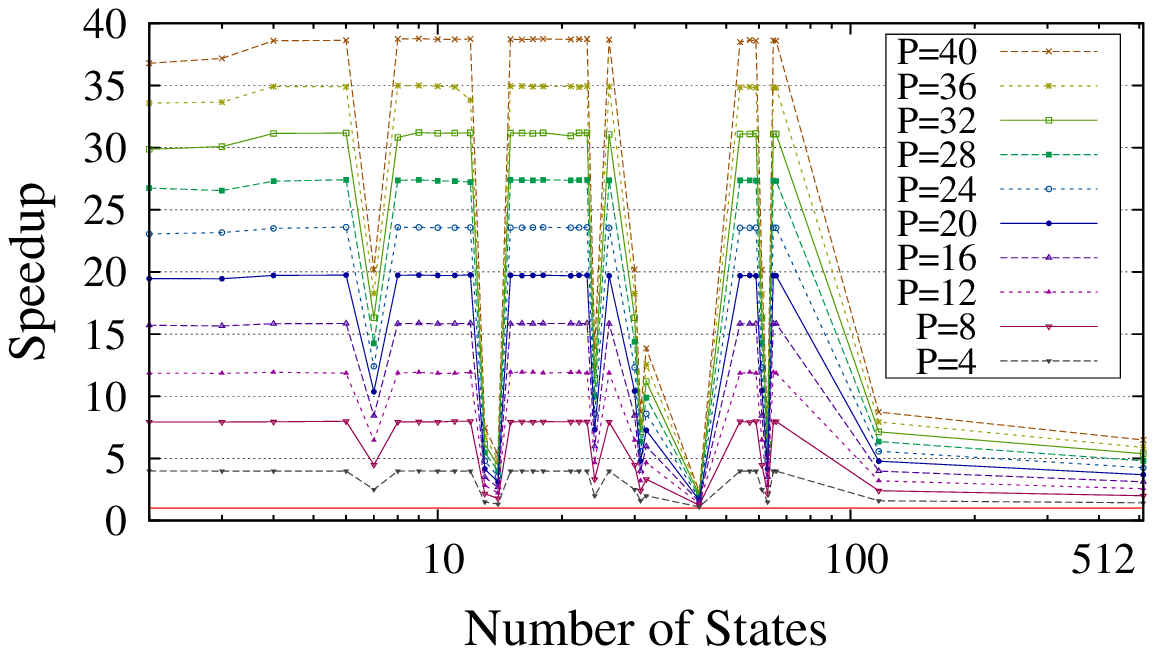}}
        &\hspace{-8mm}{
        \subfigure[PCRE speedups of \MaxNrIStates\ opt.]
            {\label{fig:speedup_mtl_opt_against_noOpt_re}
             \includeGraphics[clip=true, height=4.3cm, trim=12mm 6mm 0 0]
                             {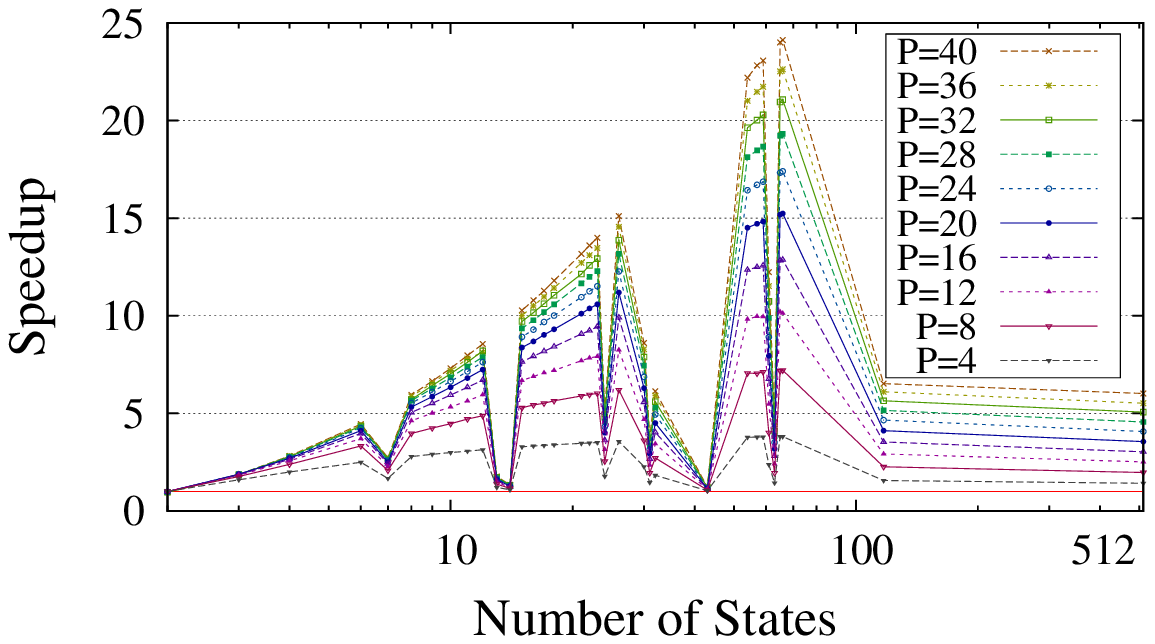}}}
        \\
    \end{tabular}
    \caption{Speedups on a shared-memory 4 CPU (40 cores) node of the Intel MTL 
        for PROSITE patterns~$(a)$ and PCRE~$(c)$, including the \MaxNrIStates\
        optimization for four symbols reverse lookahead.
        Fig.~\ref{fig:speedup_mtl_opt_against_noOpt_prosite} and 
        Fig.~\ref{fig:speedup_mtl_opt_against_noOpt_re} show speedups 
        achieved by the \MaxNrIStates\ optimization over matching for \NrStates.}
    \label{fig:speedup_mtl_par}
\end{figure*}

Fig.~\ref{fig:speedup_mtl_par} shows the results of
our speculative parallel DFA membership test
with and without applying four symbol reverse lookahead, for the PROSITE and PCRE benchmark
suites conducted on the Intel MTL. We used GCC~4.5.1 on RedHat RHEL~5.4
(x86\_64 kernel version~2.6.18-164.el5).
$x$-axes denote the number of states~\NrStates, and $y$-axes
denote the speedup over sequential matching.
We note the following observations: (1)~Our algorithms
always show better performance than sequential matching, despite  
the overhead
from redundant computations incurred by speculative parallelization.
Redundant computations constitute matching of subsequent chunks for multiple DFA states, which contrasts
sequential DFA matching where the input is only matched for the start state~\StartState.
(The red horizontal 
lines denote the break-even point where the speedup over sequential
matching is 1.) The fact that there are no speed-downs validates
the failure-freedom of our speculative parallelization.
(2)~Speedups are always proportional to 
\NrProcessors, as predicted by the complexity analysis
in Section~\ref{sec:TimeComplexity}. This proves our basic assumption that
the number of symbols
to be processed per processor decides the overall matching time despite the 
overhead due to parallelization.
The performance
improvements due to our \MaxNrIStates\ optimization are shown in 
Fig.~\ref{fig:speedup_mtl_opt_against_noOpt_prosite} and 
Fig.~\ref{fig:speedup_mtl_opt_against_noOpt_re}.

This result compares favorable to an approach presented
in~\cite{Holub:2009}, which has a complexity of $\BigO(\frac{n\cdot\NrStates}{\NrProcessors})$
and thus achieves speedups only if the number of processors is larger than the number of
states. We evaluated the approach from~\cite{Holub:2009} for both PCRE and the PROSITE patterns,
as depicted in
Fig.~\ref{fig:speedup_mtl_holub}.
In-line with the algorithm's complexity results,
the previous approach cannot achieve speedups when 
$\NrProcessors\leq\NrStates$.
We observed an almost 390x  speed-down for a DFA with 788 states.
In contrast, our algorithm achieved a speedup between 2.3x and 38.8x for
PCRE, and between 1.3x and 19.9x for PROSITE.

\begin{figure}[ht]
\centering
    \begin{tabular}{@{}c@{}c@{}}
        \subfigure[PROSITE protein patterns]
            {\label{fig:speedup_mtl_holub_prosite}
             \hspace{-2.2mm}
             \includeGraphics[clip=true, height=4.2cm, trim=4mm 6mm 0 0]
                             {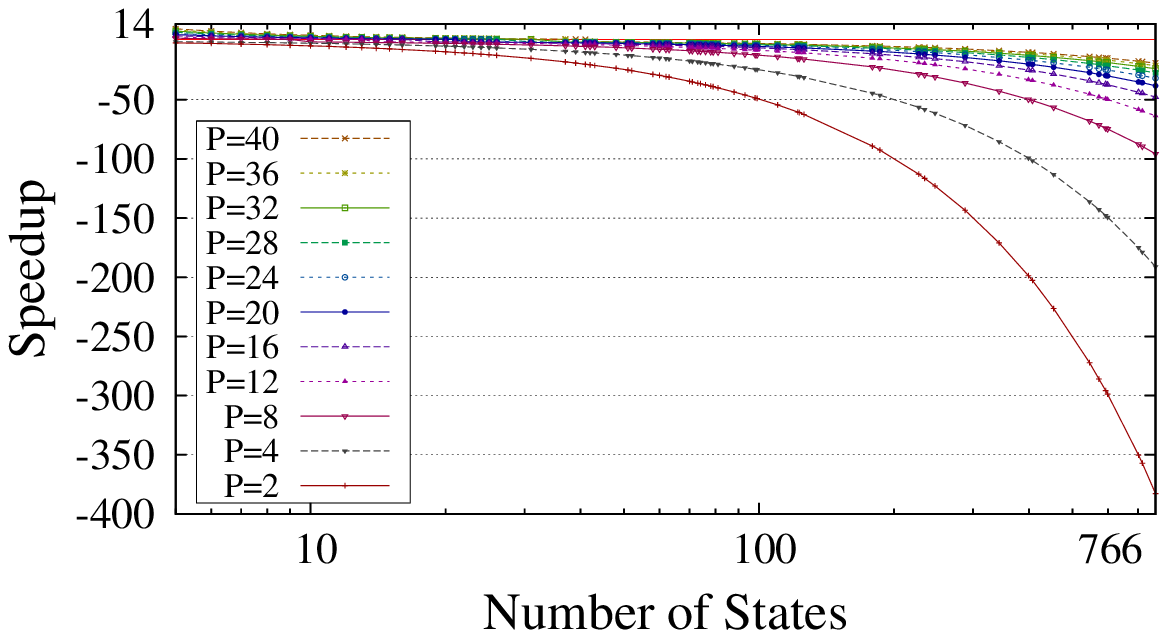}}
        &\hspace{-7mm}{
        \subfigure[PCRE regular expressions]
            {\label{fig:speedup_mtl_holub_re}
             \includeGraphics[clip=true, height=4.2cm, trim=10mm 6mm 0 0]
                             {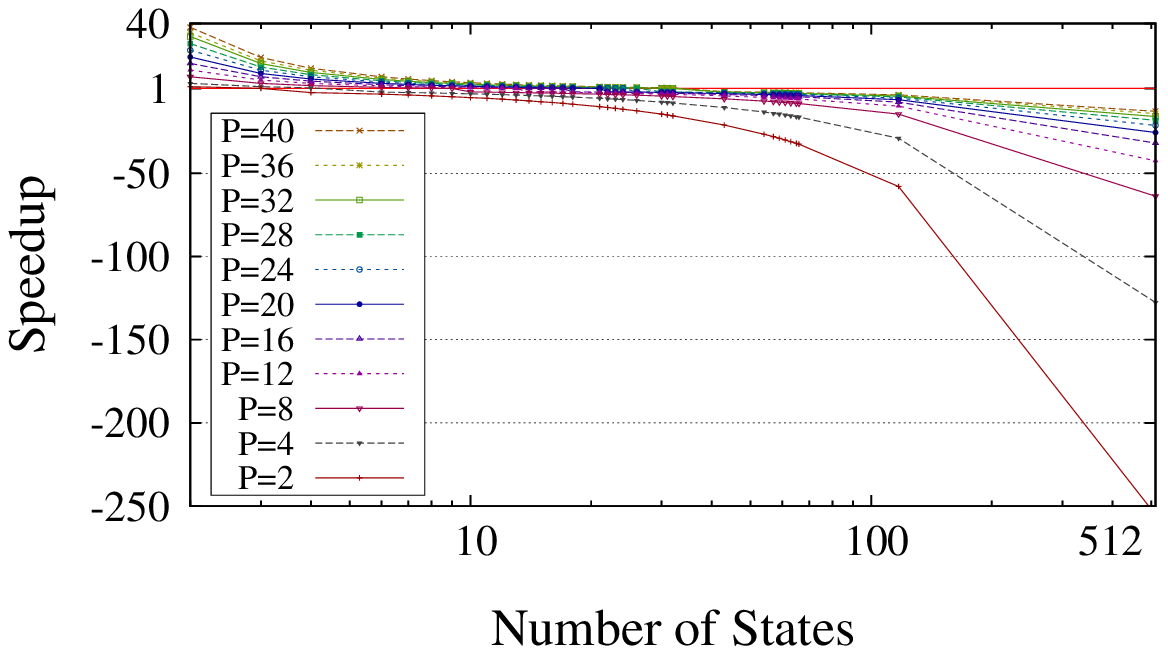}}}
        \vspace{-2mm}\\
    \end{tabular}
    \caption{Performance of the approach from \cite{Holub:2009} on the Intel MTL.}
    \label{fig:speedup_mtl_holub}
\end{figure}

\begin{figure*}[ht]
\centering
    \begin{tabular}{@{}c@{}c@{}}
        \subfigure[Speedup over ScanProsite]
            {\label{fig:speedup_mtl_prositescan_perl}
             \hspace{-2.2mm}
             \includeGraphics[clip=true, height=4.2cm, trim=4mm 7mm 0 0]
                             {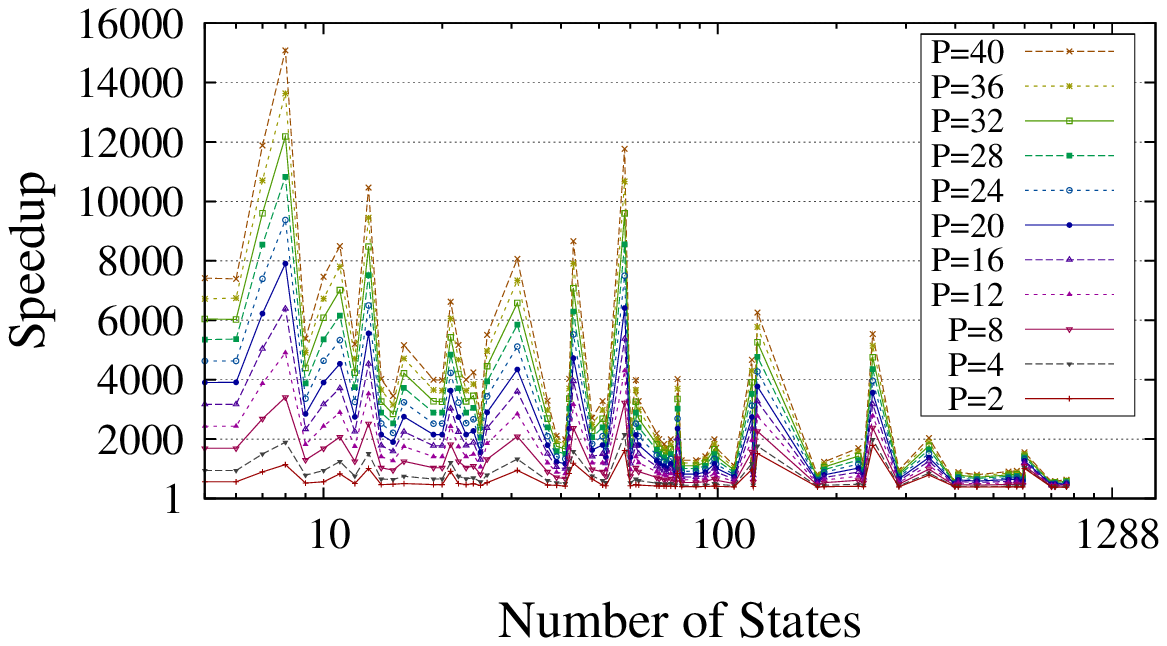}}
        &\hspace{-7mm}{
        \subfigure[Speedup over \texttt{grep}]
            {\label{fig:speedup_mtl_prositescan_grep}
             \includeGraphics[clip=true, height=4.2cm, trim=10mm 7mm 0 0]
                             {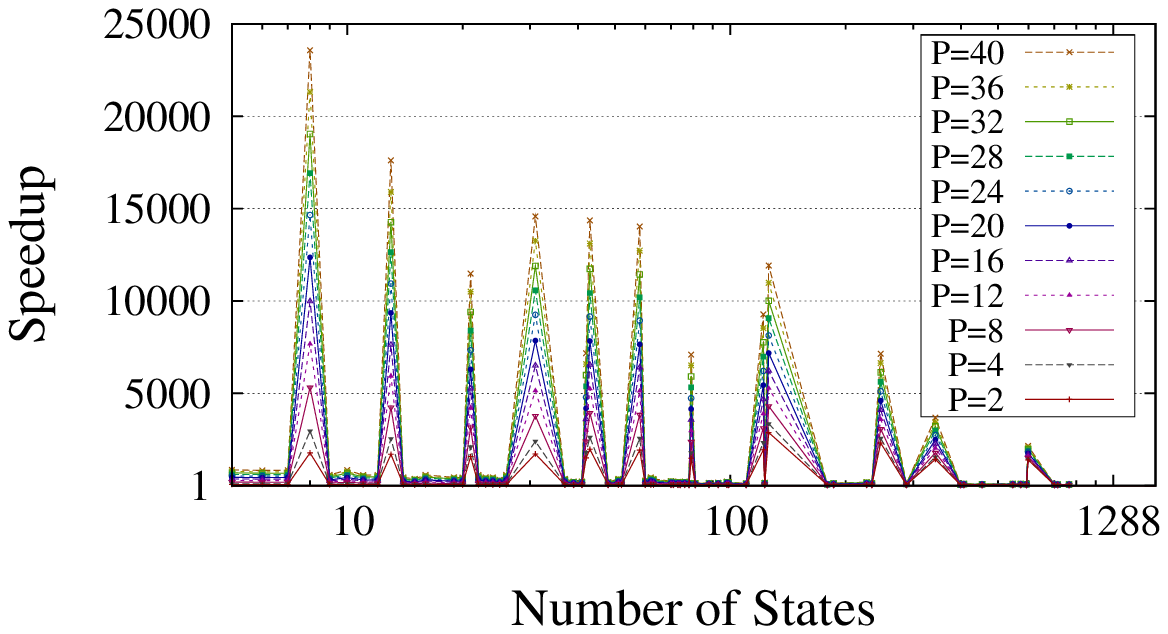}}}
        \vspace{-2mm}\\
    \end{tabular}
    \caption{Performance of our approach compared to ScanProsite~(a) and the UNIX \texttt{grep} utility~(b).}
    \label{fig:speedup_mtl_prositescan}
\end{figure*}

Another experiment conducted on the MTL is
the comparison to ScanProsite~\cite{ScanProsite:2002,ScanProsite},
which is the reference implementation
from the PROSITE protein database. ScanProsite is used
to detect signature matches in protein sequences.
The tool is implemented in Perl; it can be used
to find all substrings that match a certain
PROSITE pattern. We parameterized ScanProsite to find only
one match to compare with our optimized DFA matching algorithm which determines
whether an input string contains a certain pattern or not.
For a second comparison, we employed the UNIX grep utility with ScanProsite.
Grep constructs a DFA and uses the Boyer-Moore algorithm for matching~\cite{haertel2010};
it is faster than Perl which uses backtracking~\cite{cox2007}.
As shown in Fig.~\ref{fig:speedup_mtl_prositescan}, 
our algorithm using four symbol reverse lookahead is 559.3 to 15079.7 times faster 
than ScanProsite, and 
62.1 to 23572.0 times faster than the UNIX \texttt{grep} utility.

\subsection{Performance of Vectorized DFA Matching Using AVX2 Instruction Set Extensions}
\label{subsec:EvalSimd}
At the time of writing, CPUs supporting AVX2 instruction set extensions were
not commercially available. To validate our vectorized code, and to get an
indication on
the speedups obtainable with AVX2, we resorted to
Intel's SDE emulator~\cite{SDE}, version~4.46.0.  
SDE is not cycle-accurate, but it provides the number of machine
instructions ``executed'' by the emulated binary. We used
the number of executed machine instructions as the basis of
our performance comparison of scalar and vectorized code.
{\em Speedup\/} throughout Section~\ref{subsec:EvalSimd}
thus denotes a ratio of executed machine instructions, as
opposed to observed execution time.
Performance on real hardware, e.g., Intel's Haswell microarchitecture, must be
expected to vary to the extent of variations in the cycles per instruction
(CPI) between scalar and vectorized code.
For compilation of code with AVX2 intrinsics,
we used ICC version~12.1.4.

\begin{figure*}[ht]
\centering
    \begin{tabular}{@{}c@{}c@{}}
        \subfigure[Speedup w/o AVX2 for PROSITE]
            {\label{fig:speedup_par_opt_prosite}
             \hspace{-3.2mm}
             \includeGraphics[clip=true, height=4.3cm, trim=4mm 6mm 0 0]
                             {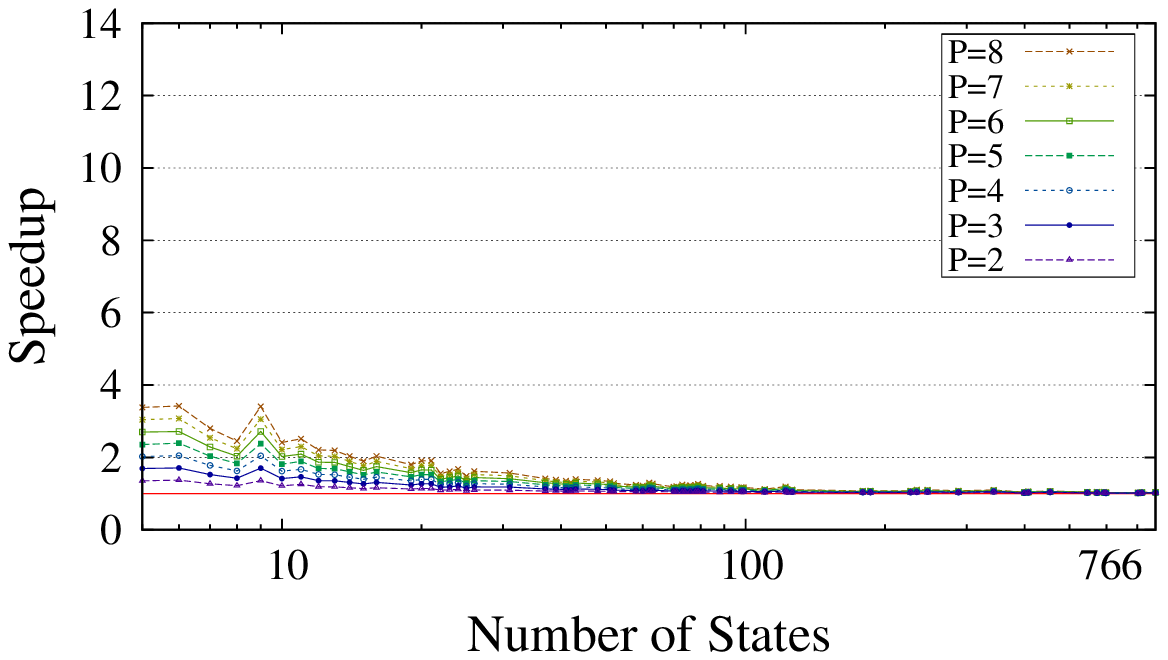}}
        &\hspace{-5.5mm}{
        \subfigure[Speedup w/ AVX2 for PROSITE]
            {\label{fig:speedup_simd_opt_prosite}
             \includeGraphics[clip=true, height=4.3cm, trim=12mm 6mm 0 0]
                             {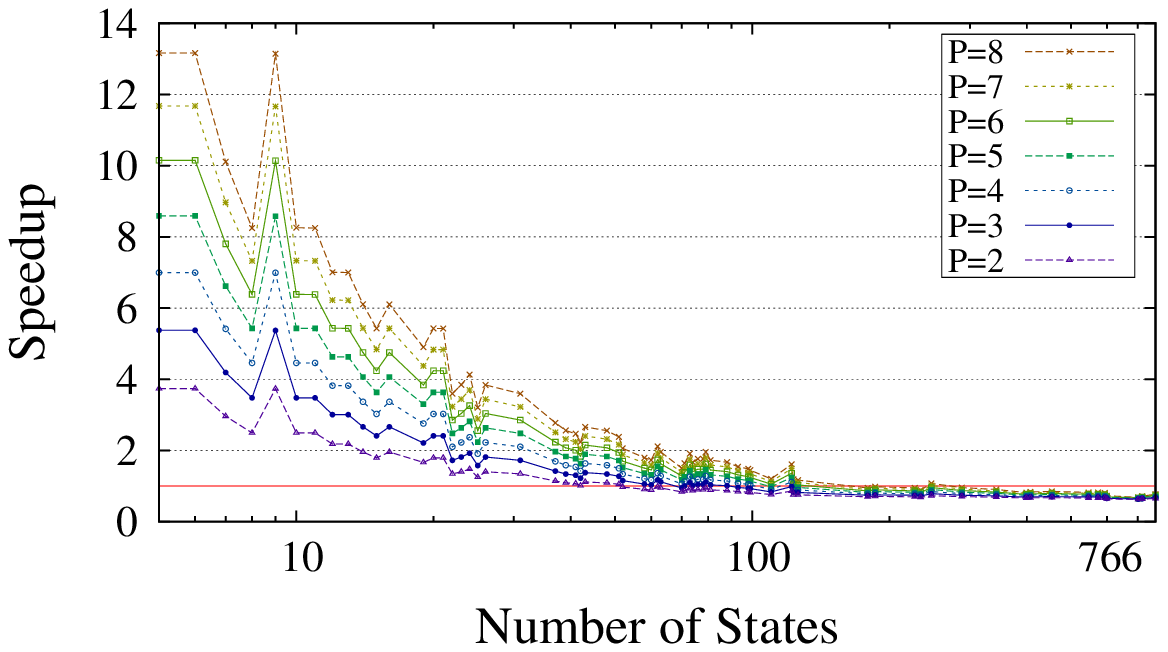}}
        }
        \vspace{-5mm}\\
        \subfigure[Speedup w/o AVX2 for PCRE]
            {\label{fig:speedup_par_opt_re}
             \hspace{-3.2mm}
             \includeGraphics[clip=true, height=4.3cm, trim=4mm 6mm 0 0]
                             {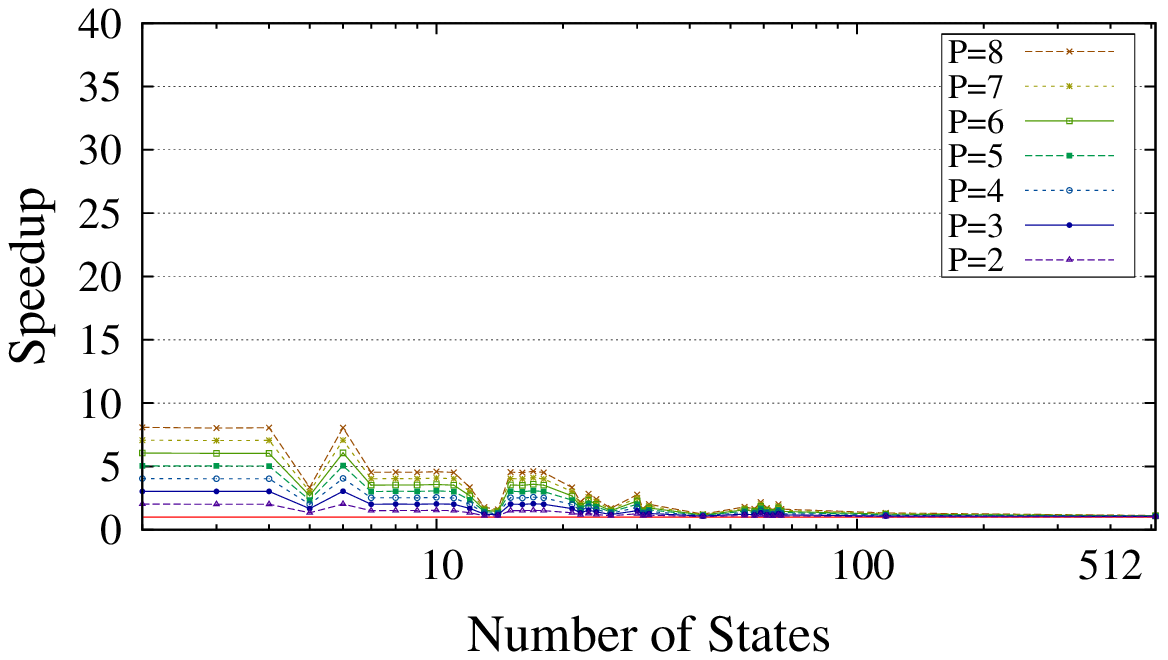}}
        &\hspace{-7.8mm}{
        \subfigure[Speedup w/ AVX2 for PCRE]
            {\label{fig:speedup_simd_opt_re}
             \includeGraphics[clip=true, height=4.3cm, trim=12mm 6mm 0 0]
                             {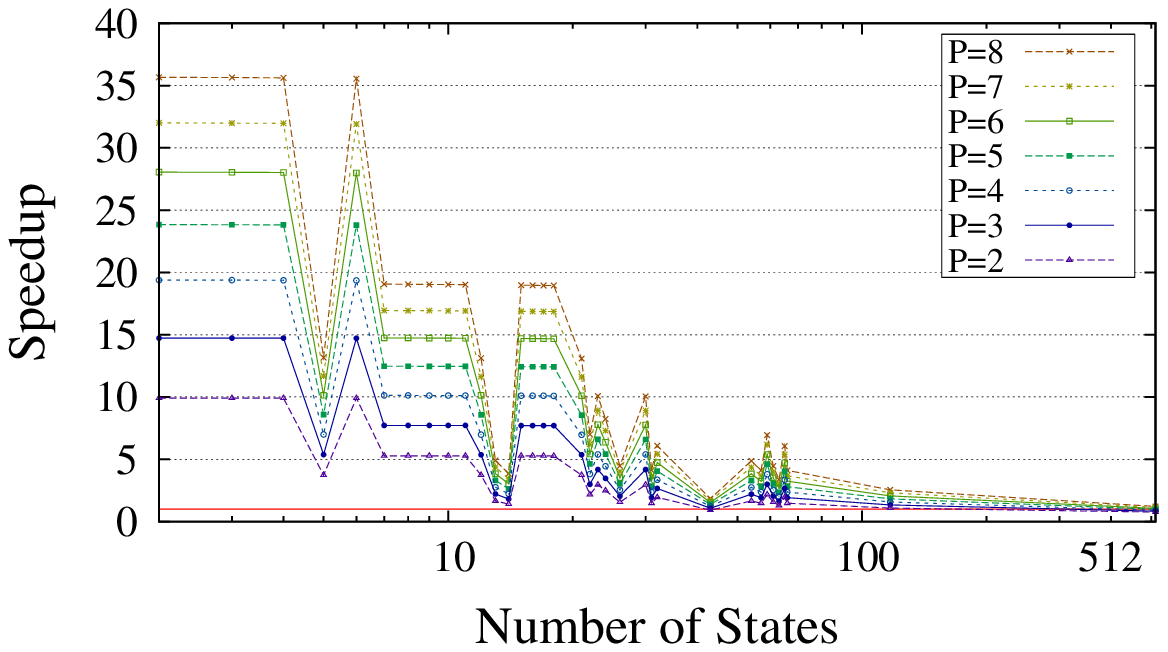}}}
        \\
    \end{tabular}
    \caption{Speedups from AVX2 8-fold vectorization for PROSITE patterns~$(b)$ and 
        PCRE~$(d)$. Fig.~\ref{fig:speedup_par_opt_prosite} and 
        Fig.~\ref{fig:speedup_par_opt_re} depict the achieved speedups 
        without vectorization.}
    \label{fig:speedup_simd}
\end{figure*}

Fig.~\ref{fig:speedup_simd} compares the speedups of scalar and
vectorized DFA membership tests
applying the \MaxNrIStates\ optimizations for one symbol lookahead.
Eight-fold vectorization using AVX2 instructions achieved
a 4.45x improvement over scalar code.
Furthermore, we observed that the expected speedups
on an emulated 8-core machine with AVX2 are on the order of magnitude 
of a 40-core node of the MTL. The speedups range from 
1.2x to 35.7x for PCRE and 0.7x to 13.2x for PROSITE. Speedup is again 
proportional to \NrProcessors, showing that vectorization is in-line
with our complexity analysis from Section~\ref{sec:TimeComplexity}.
We observed a 16.0\% speed-down on average~(maximum 31.5\%) with 
very large DFAs due to the overhead of our  
parallelization for SIMD operations. This speed-down is 
not innate to the algorithms, but due to our implementation, in particular
the way chunks are allocated to SIMD vector units. The speed-down can be
overcome by increasing the problem size (which we refrained from, to keep
experiments consistent).

\subsection{DFA Matching Performance on Cloud Computing Architectures}
\label{subsec:EvalMPI}

\begin{figure*}[ht]
\centering
    \begin{tabular}{@{}c@{}c@{}}
        \subfigure[PROSITE speedups on EC2]
            {\label{fig:speedup_mpi_par_opt_prosite}
             \hspace{-2.2mm}
             \includeGraphics[clip=true, height=4.3cm, trim=4mm 6mm 0 0]
                             {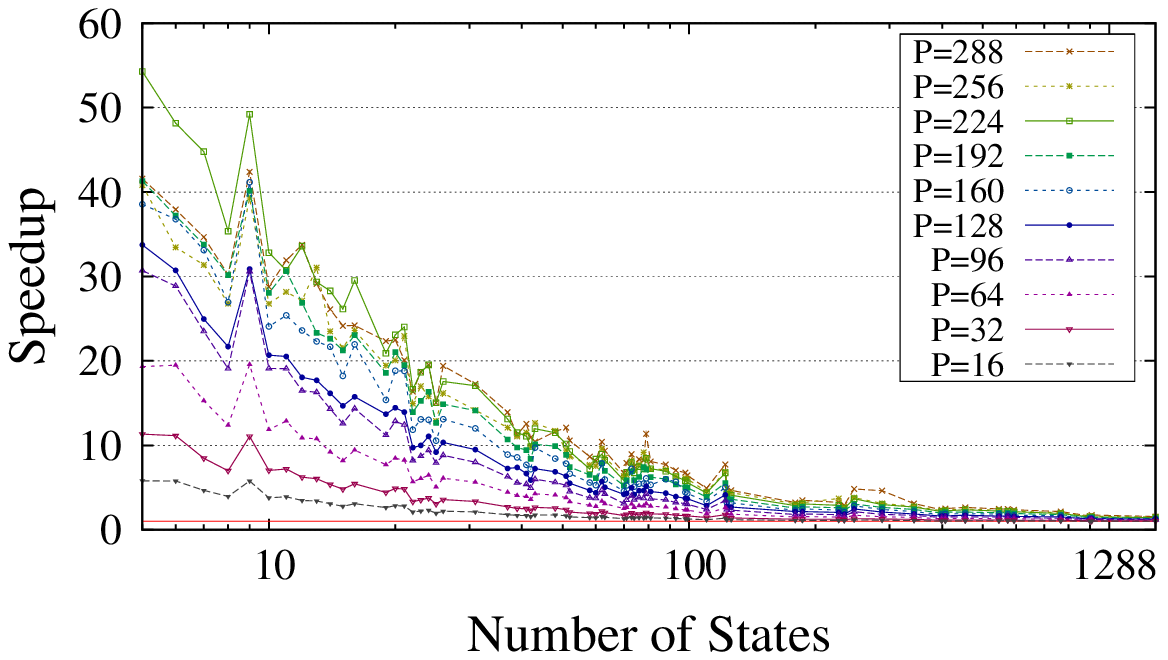}}
        &\hspace{-7mm}{
        \subfigure[PROSITE communication overhead]
            {\label{fig:speedup_mpi_par_opt_prosite_merging}
             \includeGraphics[clip=true, height=4.35cm, trim=12mm 6mm 0 0]
                             {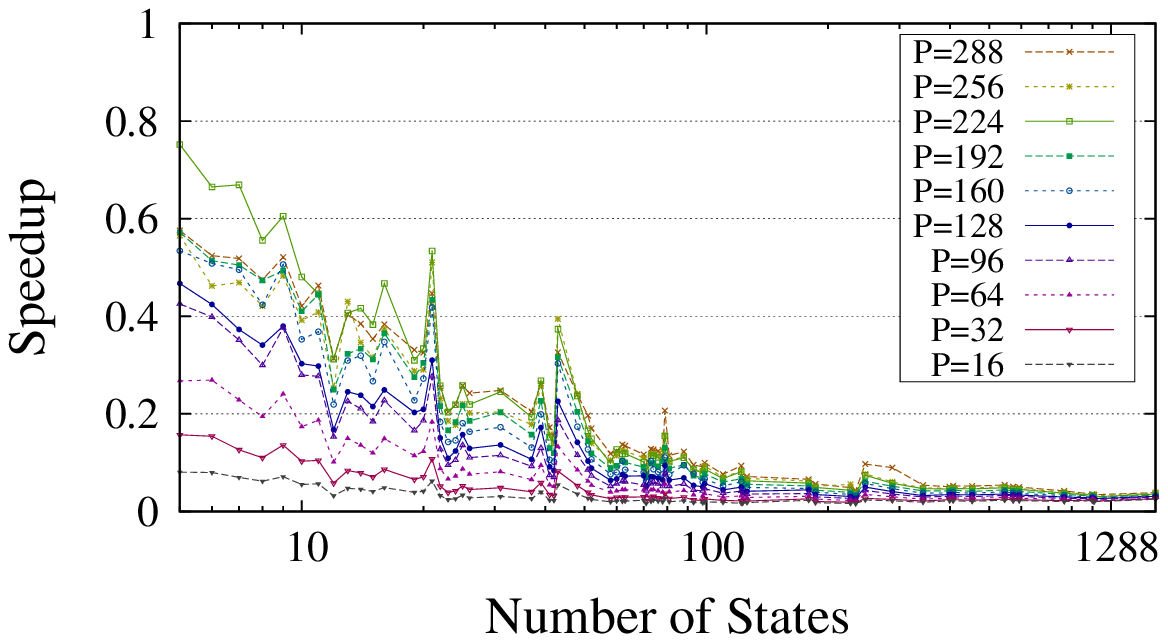}}}
        \vspace{-4mm}\\
        \subfigure[PCRE speedups on EC2]
            {\label{fig:speedup_mpi_par_opt_re}
             \hspace{-2.2mm}
             \includeGraphics[clip=true, height=4.3cm, trim=4mm 6mm 0 0]
                             {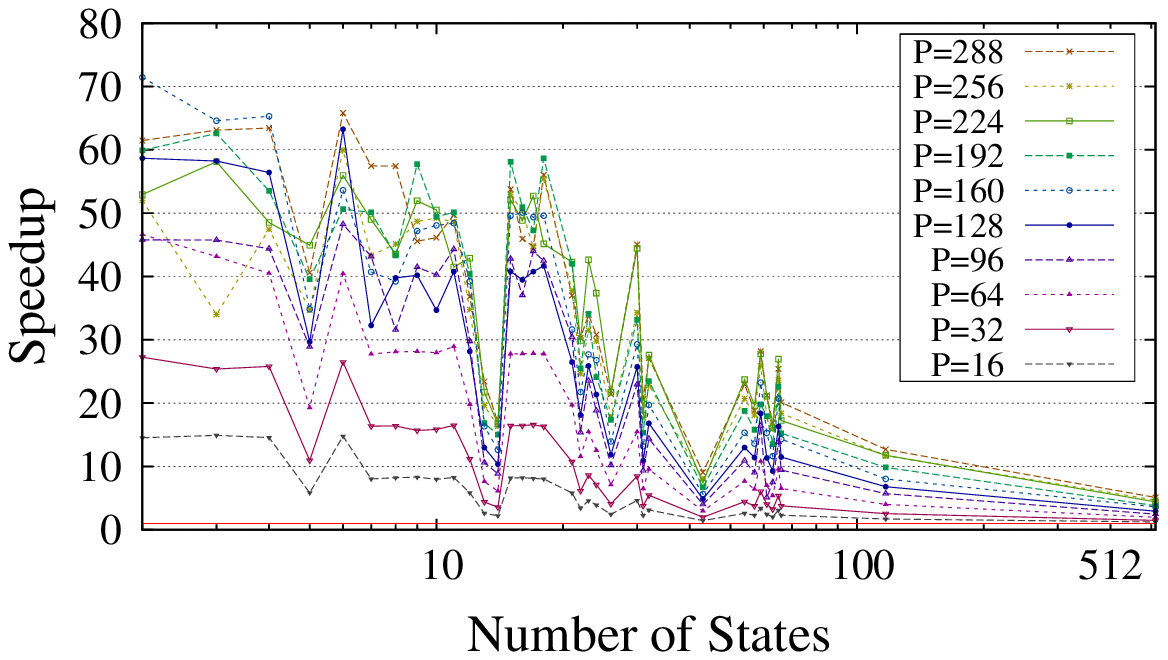}}
        &\hspace{-7mm}{
        \subfigure[PCRE communication overhead]
            {\label{fig:speedup_mpi_par_opt_re_merging}
             \includeGraphics[clip=true, height=4.35cm, trim=12mm 6mm 0 0]
                             {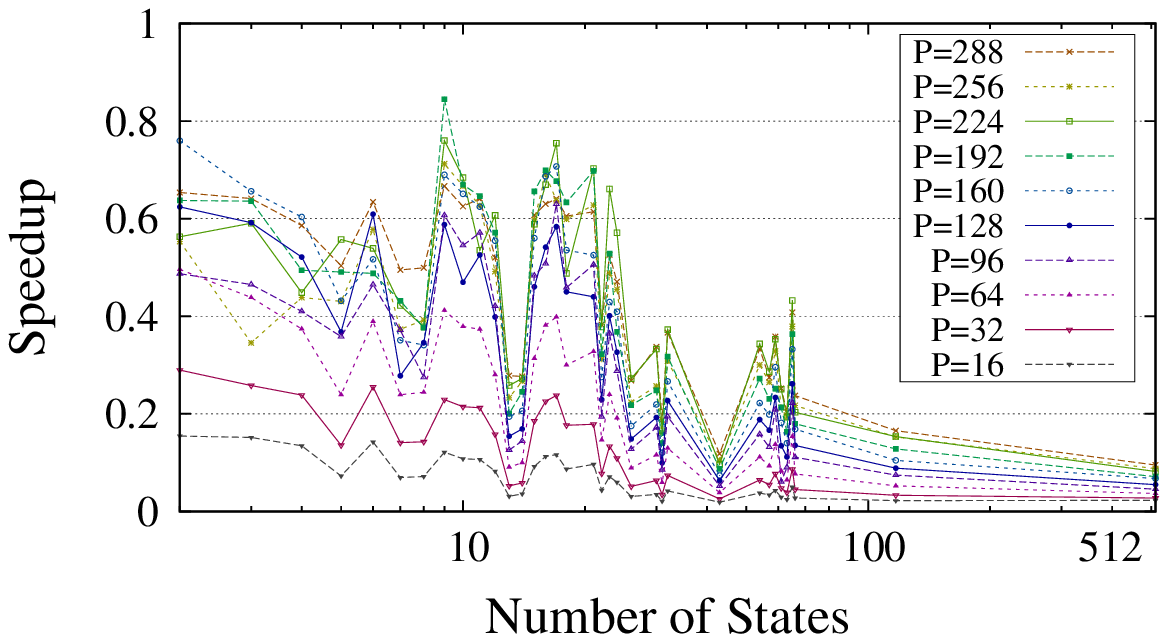}}}
        \\
    \end{tabular}
    \caption{Performance on cloud computers~(cc2.8xlarge instance type on EC2) with 
        PROSITE patterns~$(a)$ and PCRE~($c$). 
        Fig.~\ref{fig:speedup_mpi_par_opt_prosite_merging} and 
        Fig.~\ref{fig:speedup_mpi_par_opt_re_merging} show proportional overhead
        of MPI communication cost with PROSITE patterns and PCRE respectively.}
    \label{fig:speedup_mpi_par}
\end{figure*}

We conducted experiments on the Amazon EC2 elastic computing cloud
to determine the performance of our speculative DFA matching algorithms
on distributed-memory architectures, employing up to 20~nodes and 288~cores.
We explored the adaptation of our load-balancing approach to
EC2 nodes of varying processing capacities.
For the convenience of operating a cluster of EC2 nodes,
we used StarCluster~\cite{StarCluster} version~0.93.3, 
an open source cluster-computing 
toolkit for EC2.

Experiments were conducted on up to 20~cc2.8xlarge EC2 instances, which
provide 16~cores per node.
We again employed four symbols reverse lookahead with our approach.
For reasons discussed in Section~\ref{sec:cloudImplementation},
we occupied 15 out of 16~cores,
resulting 
in 300~cores in total. For better presentation, 
Fig.~\ref{fig:speedup_mpi_par} shows our experimental results with
cluster sizes that are a multiple
of 32~cores. 
We used version~1.4
of the MPICH2 MPI implementation~\cite{MPICH2}, which provided 
higher performance on EC2 than
OpenMPI~\cite{OpenMPI} version~1.4.3. 

We found the communication costs between nodes an important factor
on the EC2 cloud.
We instrumented our matching framework to determine the communication
overhead. 
Fig.~\ref{fig:speedup_mpi_par_opt_prosite} and 
Fig.~\ref{fig:speedup_mpi_par_opt_re} show speedups including communication
costs, and
Fig.~\ref{fig:speedup_mpi_par_opt_prosite_merging} and 
Fig.~\ref{fig:speedup_mpi_par_opt_re_merging} 
depict the ratio of time spent for communication 
to overall execution time.
Although graphs shown in 
Fig.~\ref{fig:speedup_mpi_par} are irregular due to the instability
of the EC2 network, we can observe that the
communication costs increase as the number of processors grows. The
communication cost decreases as \NrStates\ grows, which follows
from the fact that
the required matching time increases with \NrStates, which de-emphasizes
communication costs. This observation explains why PCRE benchmarks,
which show smaller $\MaxNrIStates$ constants and smaller sets of possible
initial states, are more impacted by
communication overhead 
than PROSITE benchmarks.

\begin{table*}[htbp]
\centering
\begin{tabular}{|c|c||c|c|c||c|c|c|}
\hline
\multicolumn{2}{|c||}{EC2 Instances} &
    \multicolumn{3}{c||}{PROSITE}&\multicolumn{3}{c|}{PCRE}\\
    \hline
    Fast& Slow& Min. & Avg. & Max. & Min. & Avg. & Max. \\
    \hline
    \hline
0 & 5 & 0.0036 & 0.0102 & 0.0298 & 0.0046 & 0.0149 & 0.0696 \\
1 & 4 & 0.0031 & 0.0086 & 0.0360 & 0.0036 & 0.0108 & 0.0355 \\
2 & 3 & 0.0033 & 0.0090 & 0.0275 & 0.0062 & 0.0121 & 0.0427 \\
3 & 2 & 0.0051 & 0.0116 & 0.0248 & 0.0083 & 0.0186 & 0.0707 \\
4 & 1 & 0.0060 & 0.0130 & 0.0700 & 0.0093 & 0.0194 & 0.0707 \\
5 & 0 & 0.0056 & 0.0119 & 0.0305 & 0.0095 & 0.0188 & 0.0412 \\
  \hline
\end{tabular}
\caption{Effectiveness of the load-balancing scheme on six configurations
of inhomogeneous clusters consisting of
two types of Amazon EC2 instances, m2.4xlarge and cc2.8xlarge}
\label{tab:weightFeasibility}
\end{table*}

The goal of our load-balancing mechanism is to determine chunk sizes such that
all processing cores are utilized equally, i.e., take equally long for matching their assigned chunk.
Processor capacities are incorporated in the form
of weights (see Eq.~\eqref{eq:WeightedL}).
To evaluate the load-balance achieved with our speculative DFA matching computations,
we used two different types of Amazon EC2 instances, namely cc2.8xlarge
(denoted as ``Fast'' in Table~\ref{tab:weightFeasibility}) and
m2.4xlarge (denoted as ``Slow'' in the second column of the table).
Although the clock frequencies of these EC2 instance types
do not differ much (see Table~\ref{tab:HWSpec}), the difference of
the processor capacities is observable.
We found the ratio of actual processing capacity of 
cc2.8xlarge compared to m2.4xlarge to be 1.41 on average, meaning that cc2.8xlarge
on average computes 41\% faster than m2.4xlarge.
For this experiment, we allocated inhomogeneous clusters consisting
of various numbers of cc2.8xlarge and m2.4xlarge instances.
To get an indication for the effectiveness of our load-balancing scheme,
we determined the 
standard-deviations of DFA matching times across all cores of such inhomogeneous
EC2 clusters.
A balanced load would then be indicated by standard deviations close to zero.
E.g., the experiment from row~5 was conducted on a mix
of 
four cc2.8xlarge instances and one m2.4xlarge 
instance. The maximum observed standard deviation of execution times
was 7.0\%, with 0.6\% minimum standard deviation and 1.3\% average across the
PROSITE benchmark suite. 
During experiments, we noticed that capacities of cluster nodes could change
slightly across cluster invocations, making the re-estimation of processor capacities
necessary at cluster startup time.
(This is in line with the findings from
\cite{EC2Perf}, on performance unpredictability of cloud computing environments.)
Hence the adaptability of our load-balancing scheme with respect to processor capacities
is essential
on cloud computing environments. Our observed proportional standard 
deviations of 
execution times are very low, around 1\% on average, as shown in 
Table~\ref{tab:weightFeasibility}. In particular, the presented load-balancing scheme adapts
well to different configurations of inhomogeneous clusters.

\subsection{Performance Impact of Structural DFA Properties and Scalability to Large Input Sizes
\label{subsec:struct}}

\begin{figure}
\centering
    \begin{tabular}{@{}c@{}c@{}}
        \subfigure[PROSITE patterns]
            {\label{fig:speedup_against_q_prosite}
             \hspace{-2.7mm}
             \includeGraphics[clip=true, height=4.2cm, trim=4mm 6mm 0 0]
                             {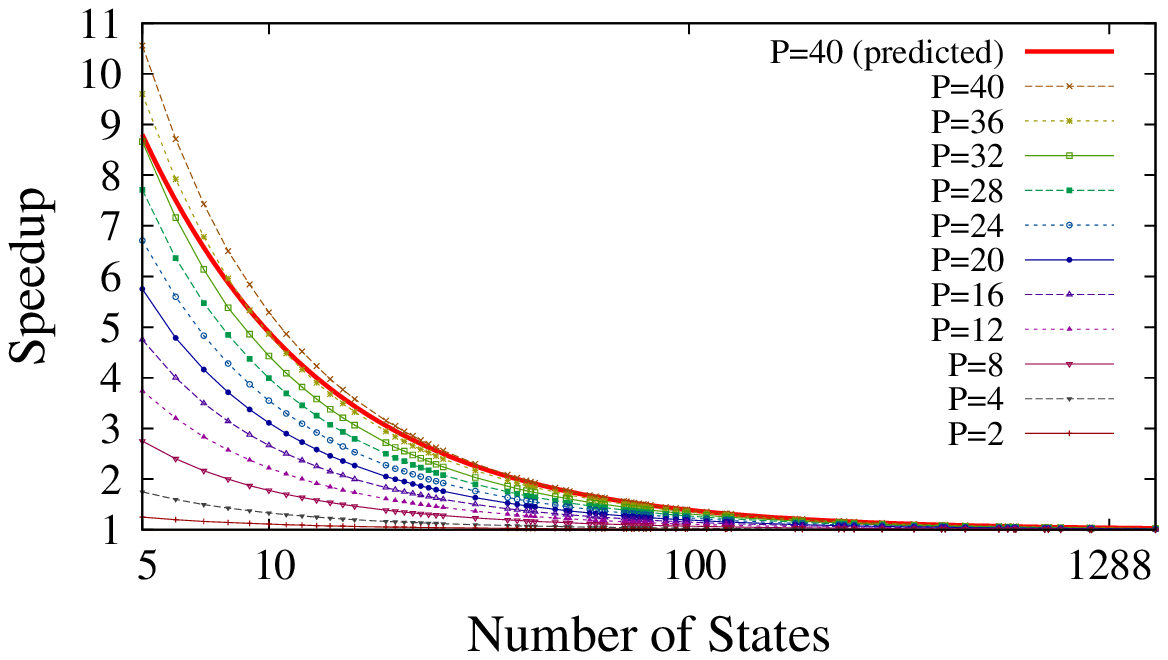}}
        &\hspace{-5mm}{
        \subfigure[PCRE]
            {\label{fig:speedup_against_q__re}
             \includeGraphics[clip=true, height=4.2cm, trim=12mm 6mm 0 0]
                             {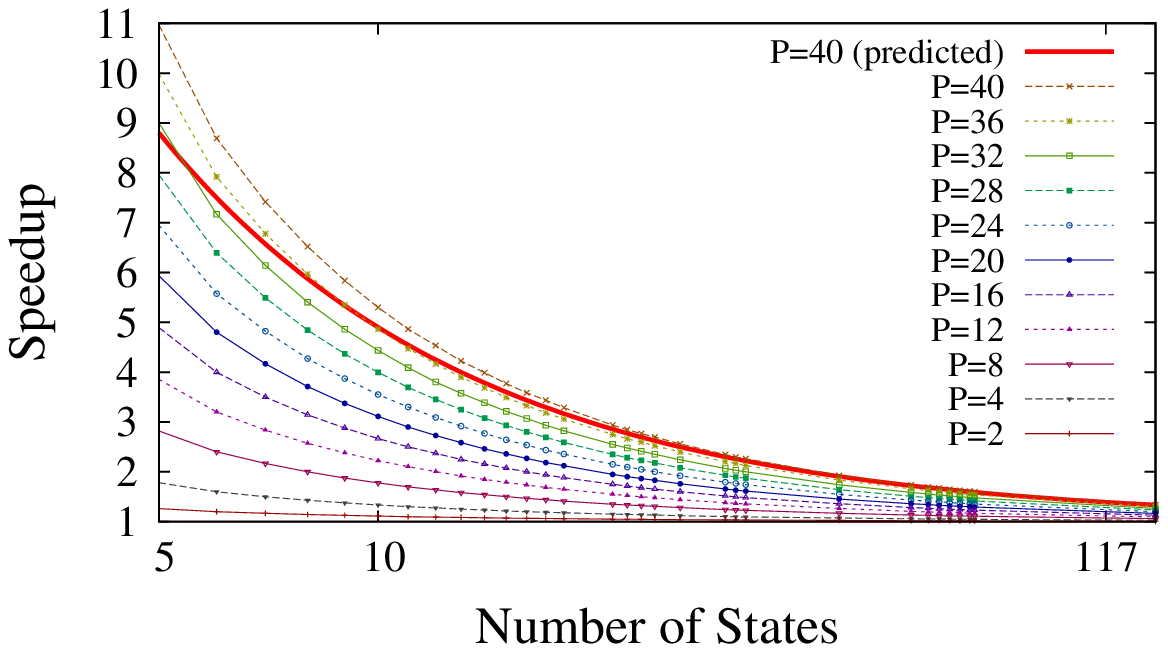}}}
        \\
    \end{tabular}
    \caption{Speedups on the Intel MTL node over sequential matching,
        without $\MaxNrIStates$ optimization. 
        Observed speedups closely follow the prediction from 
        Eq.~\eqref{eq:speedup}. The predicted
        speed-up for 40~cores is shown by
        the graph labelled ``P=40 (predicted)''.
    }
    \label{fig:speedup_against_q}
\end{figure}

If we regard the number of processing cores as a constant,
then the expected speedup of the proposed speculative DFA matching approach
solely depends on the number
of DFA states~\NrStates\ (see Eq.~\eqref{eq:speedup}).
Fig.~\ref{fig:speedup_against_q} shows
the speedups that we obtain without the $\MaxNrIStates$ optimization. The
line labeled ``Theoretical Speedup for P=40'' depicts
the expected speedup according to Eq.~\eqref{eq:speedup}. The observed speedups
closely follow this trend.
Speedups without $\MaxNrIStates$ optimization are thus very regular. This
contrasts the speedups obtained with $\MaxNrIStates$ optimization enabled,
because DFAs with the same number of states~\NrStates\ can vary drastically
in  $\MaxNrIStates$, i.e., their maximum number of
possible initial states. As a consequence, the number of DFA states~\NrStates\ 
cannot predict the performance of the  $\MaxNrIStates$ optimization. We
will thus consider to what extent different numbers of reverse lookahead
symbols reduce the number of possible initial start states.

\begin{figure*}[ht]
\centering
    \begin{tabular}{@{}c@{}c@{}}
        \subfigure[PCRE]
            {\label{fig:ISetMax_re}
             \hspace{-2.2mm}\includeGraphics[clip=true,
                                             width=5.6cm,
                                             trim=0.1mm -3.5mm 0 0]{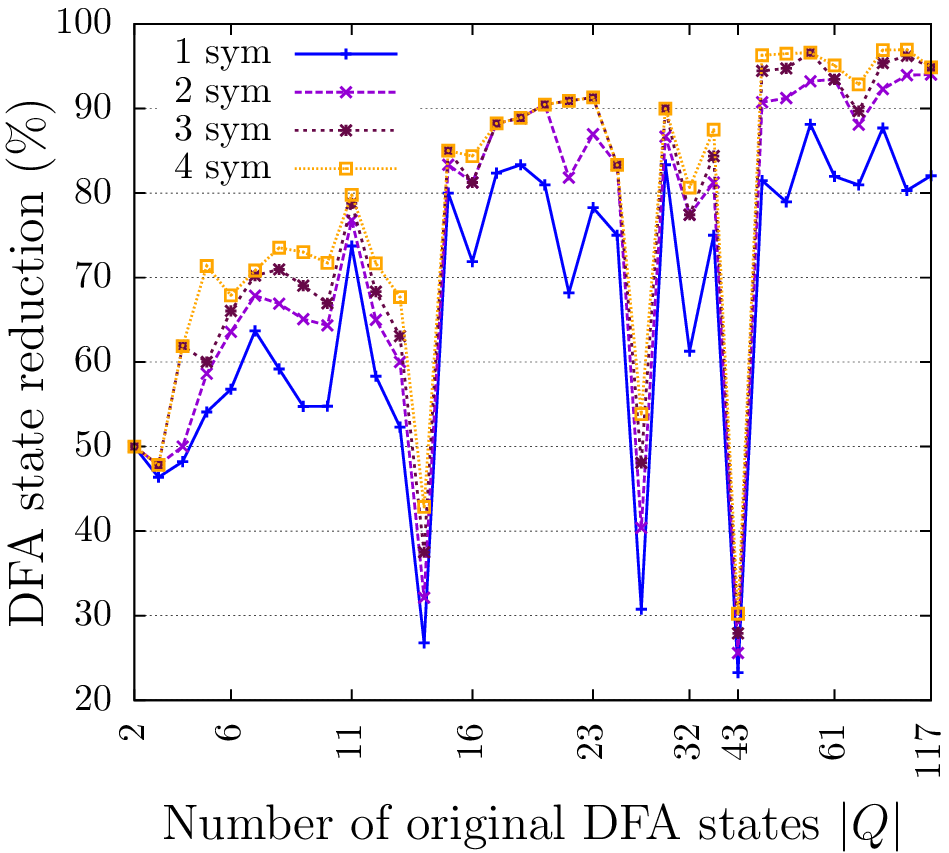}}
        &\hspace{-5mm}\raisebox{-1.54mm}{%
        \subfigure[PROSITE]
            {\label{fig:ISetMax_prosite}
             \includeGraphics[clip=true, 
                                             width=6.87cm,
                                             trim=6mm -3.5mm 0 0]{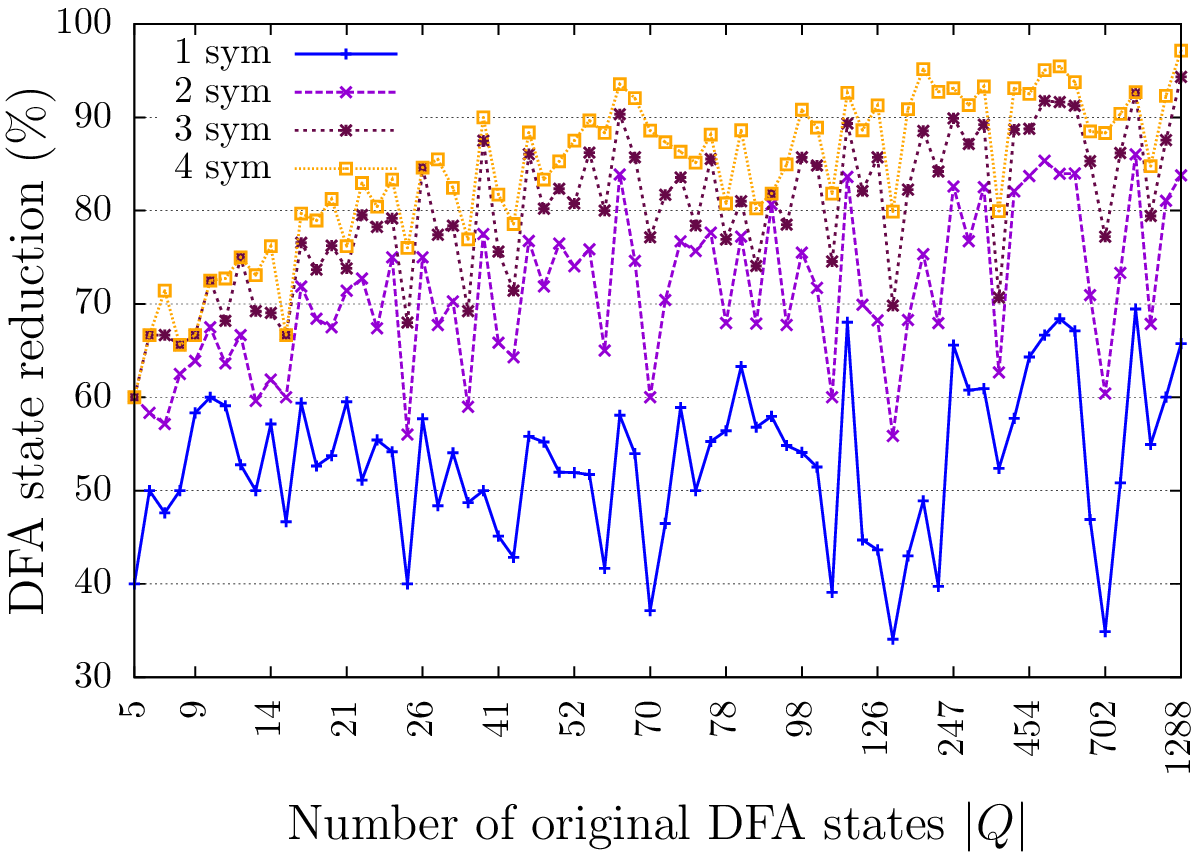}}}
        \\
    \end{tabular}
    \caption{Original sizes of DFAs and the reduction rate for various numbers of 
        reverse lookahead symbols:  $\MaxNrIStates_{,1}$,
        $\MaxNrIStates_{,2}$, $\MaxNrIStates_{,3}$ and $\MaxNrIStates_{,4}$.}
    \label{fig:ISetMax}
\end{figure*}

We investigated the sizes of possible initial state sets for the PCRE and PROSITE
benchmark suites for 1, 2, 3 and 4 reverse lookahead symbols.
Fig.~\ref{fig:ISetMax} depicts the number of states~\NrStates\ and the
initial state reduction rates for \NrPCREs~PCRE benchmark DFAs and \NrPROSITEs~PROSITE protein patterns.
(For DFAs with the same number of states, the possible initial state set sizes were averaged.)
For example, the rightmost, largest DFA in Fig.~\ref{fig:ISetMax}(b) consists of \NrStates=1288~states.
One-symbol
reverse lookahead eliminates $65$\% of \NrStates. Two-symbol, three-symbol
and four-symbol lookahead remove $83$\%, $94$\% and $97$\% of all states.
It follows from Fig.~\ref{fig:ISetMax} that DFAs exhibit variations in
their initial state reduction rate, which has an impact on matching performance. For example,
the PCRE DFA with 43~states (Fig.~\ref{fig:ISetMax}(a)) shows $\MaxNrIStates$ reduction rates below $31$\%.
The resulting impact on performance can be observed in Fig.~\ref{fig:speedup_mtl_par}(c)
with the Intel MTL node and in Fig.~\ref{fig:speedup_mpi_par}(c) for the EC2 cloud. 

The average size of possible initial state sets for 1, 2, 3 and 4 reverse lookahead symbols compared
to the overall number of states~\NrStates\
is depicted in Table~\ref{tab:ISetMax}. Applying a reverse lookahead of one symbol to the PCRE benchmarks
reduces the number of possible initial states on average to 33.7\% of the original states.
Applying 2, 3 and 4
reverse lookahead symbols
yielded further reductions of 7\%, 10\% and 12\% over~\NrStates.
With the PROSITE benchmarks, one symbol reverse lookahead reduced on average to 47.2\% of the original states.
Applying 2, 3 and 4 reverse lookahead symbols yielded further reductions of 18\%, 26\% and 31\% over~\NrStates.
The profitability of reverse lookahead is a static property of DFAs, which is reflected in this data: while
for PCRE one symbol lookahead already yields a large reduction on the number of states, lookahead $>$2 symbols
does not provide substantial improvement. However, with PROSITE, one symbol reverse lookahead provided
a smaller improvement, while reverse lookahead up to 4~symbols yielded steady gains.

\begin{table}[ht]
    \centering
    \begin{tabular}{|c|c|c|c|c|c|}
        \hline
        $r$ & 0 & 1 & 2 & 3 & 4 \\
        \hline
        \hline
        PCRE & 100\% & 33.7\% & 26.4\% & 23.7\% & 21.7\% \\
        PROSITE & 100\% & 47.2\% & 29.2\% & 20.5\% & 16.0\%  \\
        \hline
    \end{tabular}
    \caption{Average size of $\MaxNrIStates_{,r}$ compared to~\NrStates,\ 
    for $r$ reverse lookahead symbols}
    \label{tab:ISetMax}
\end{table}

Because of the exponential time complexity to compute~$\MaxNrIStates_{,r}$,
there is a trade-off between the overhead of the reverse lookahead computation
and the obtainable performance gains.
To quantify this overhead, we investigated the cost of reverse lookahead
computations on an Intel Xeon 5120 CPU.
Fig.~\ref{fig:ISetCompTime}(a) shows the overhead in microseconds
to compute~$\MaxNrIStates_{,r}$ for an example DFA of $\NrStates=5$ up to three reverse lookahead characters.
As expected, the overhead is exponential in the size of~$\Sigma$.
Fig.~\ref{fig:ISetCompTime}(b) depicts the overhead for increasing numbers of states.
Because~$\MaxNrIStates_{,r}$ is a static property of a DFA, it can be computed off-line,
and then loaded when the matching operation is performed. This way the overhead
can be avoided with DFAs that are matched many times (e.g., protein patterns from databases). 

\begin{figure*}[ht]
    \centering
    \begin{tabular}{@{}c@{}c@{}}
        \subfigure[]
            {\label{fig:ISetCompTime_per_sigma_re}
             \hspace{-2.2mm}
             \includeGraphics[clip=true, height=4.3cm, trim=4mm 6mm 0 0]
                             {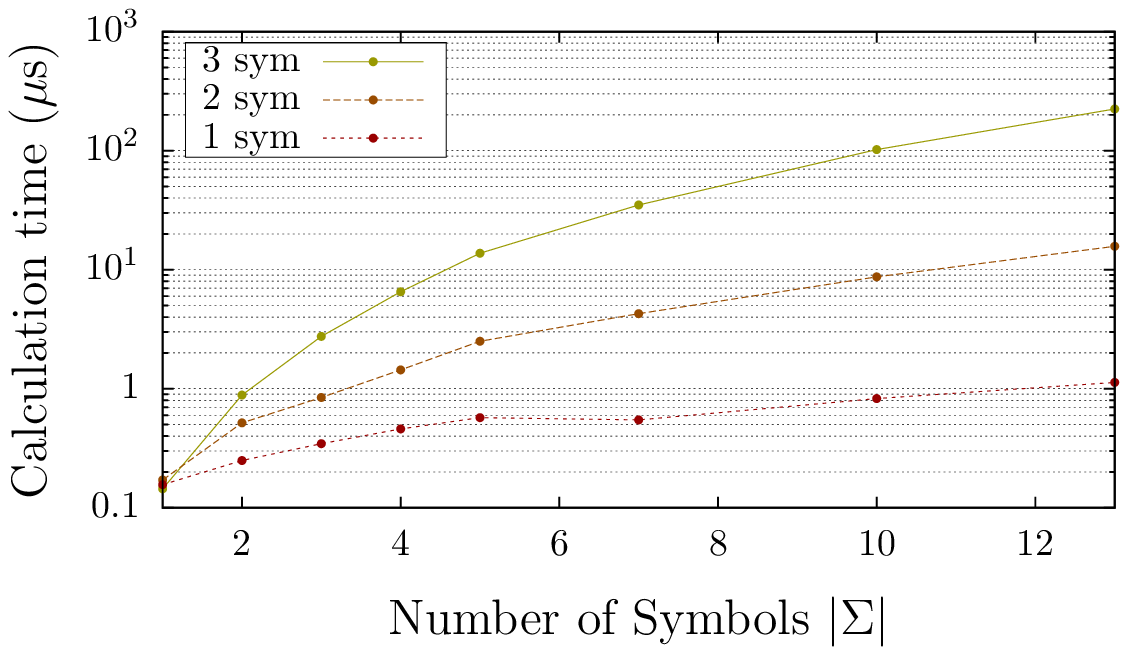}}
        &\hspace{-10mm}{
        \subfigure[]
            {\label{fig:ISetCompTime_per_Q_prosite}
             \includeGraphics[clip=true, height=4.3cm, trim=12mm 6mm 0 0]
                             {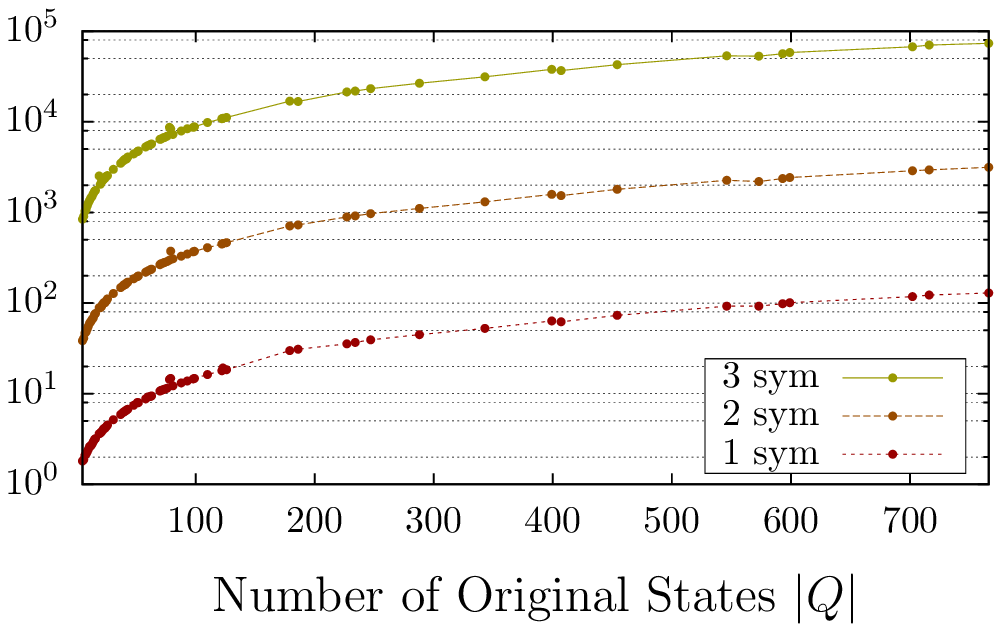}}}
        \\
    \end{tabular}
    \caption{Required overhead due to $\MaxNrIStates_{,r}$ calculation over $|\Sigma|$~$(a)$ and \NrStates~$(b)$.}
    \label{fig:ISetCompTime}
\end{figure*}

We have experimentally evaluated to what extent the size of the input string affects the DFA matching performance of our approach.
The order of magnitude of the speedups from Eq.~\eqref{eq:speedup} does not contain the size of
the input ($n$), which is reflected by the performance
data obtained on the Intel MTL node: for input sizes of \SI{1}{MB}, \SI{100}{MB}
and \SI{10}{GB}, the obtained speedups
are almost identical~(Fig.~\ref{fig:mtl_speedup_opt_prosite_per_input}). Our algorithms thus scale well with respect to large input sizes on the MTL shared memory
multicore architecture.

\begin{figure}
\centering
    \begin{tabular}{c}
    \subfigure
        {
         \includeGraphics[clip=true, height=4.3cm, trim=4mm 7mm 0 0]
                         {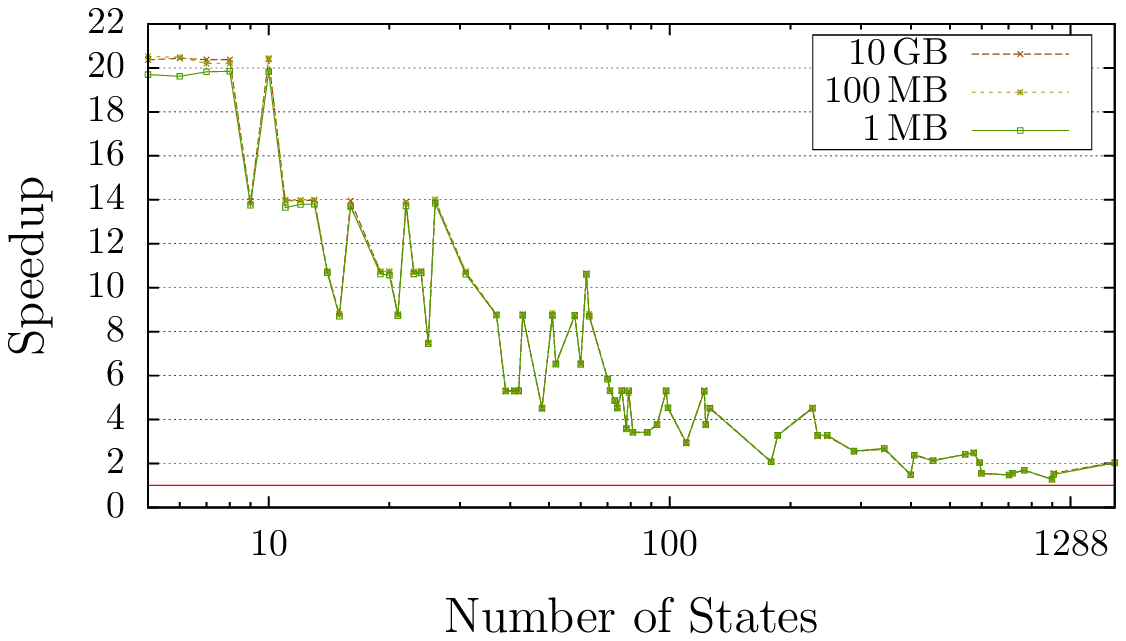}}
    \vspace{-2mm}
    \end{tabular}
    \caption{Speedups for varying input sizes and P=40 cores on the Intel MTL node.\label{fig:inputs}}
    \label{fig:mtl_speedup_opt_prosite_per_input}
\end{figure}

\begin{figure}[ht]
\centering
    \begin{tabular}{@{}c@{}c@{}}
        \subfigure[PROSITE speedup]
            {\label{fig:speedup_mpi_per_input}
             \hspace{-5.2mm}
             \includeGraphics[clip=true, height=4.3cm, trim=4mm 6mm 0 0]
                             {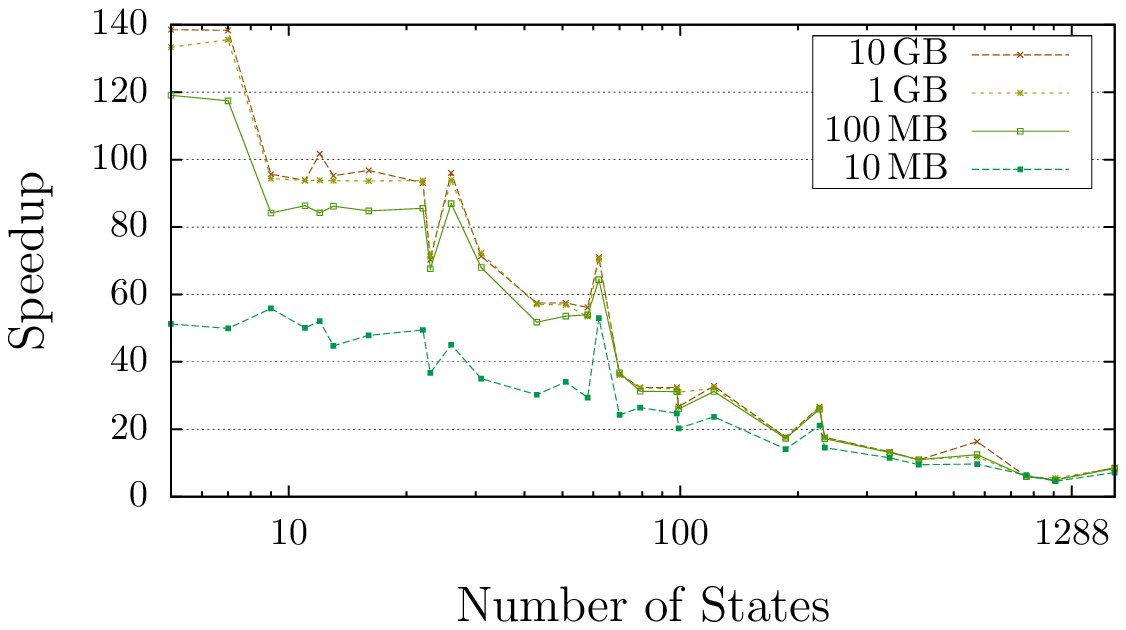}}
        &\hspace{-4.5mm}{
        \subfigure[PROSITE communication overhead]
            {\label{fig:comm_cost_mpi_per_input}
             \includeGraphics[clip=true, height=4.3cm, trim=12mm 6mm 0 0]
                             {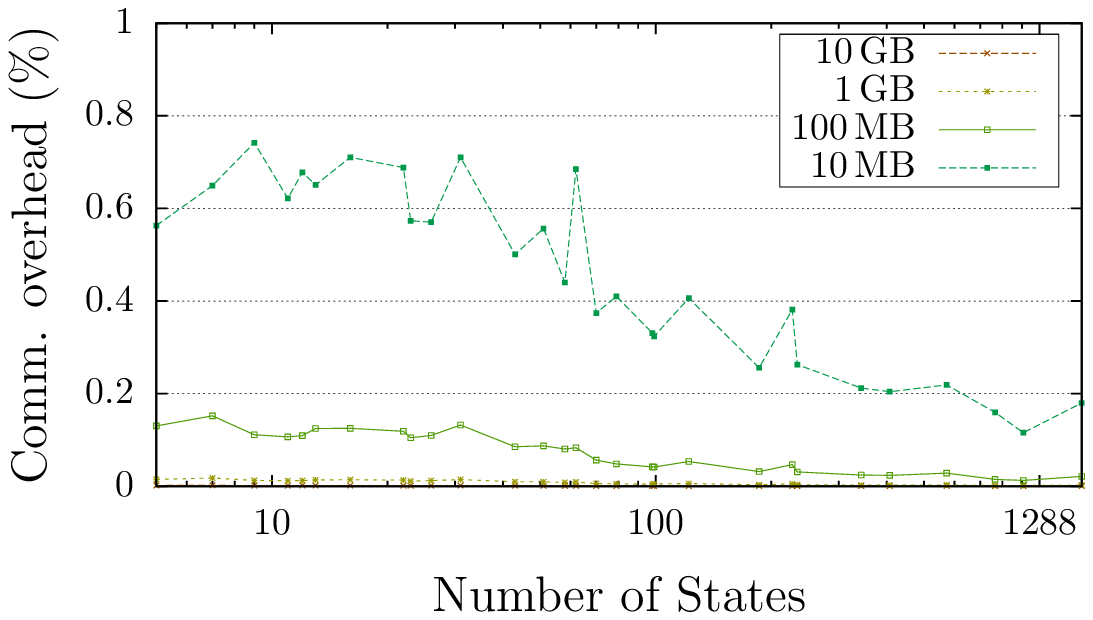}}}
        \\
    \end{tabular}
    \caption{Performance variation~$(a)$ over different sizes of inputs and its 
        proportional communication overhead~$(b)$ with PROSITE patterns on 
        cloud computers~(cc2.8xlarge instance type on EC2).}
    \label{fig:mpi_comm_cost}
\end{figure}

Fig.~\ref{fig:mpi_comm_cost}(a) depicts the performance for input sizes of
\SI{10}{MB}, \SI{100}{MB}, \SI{1}{GB} and \SI{10}{GB} for the PROSITE patterns
executed on EC2 for 288~cores.  Unlike shared memory architectures, our DFA
matching algorithm achieves
higher performance for long (\SI{10}{GB}) input. This is due to the communication
costs, which depend on the DFA size (for transmitting \StateMap-vectors with
the final state reduction), but which are independent of the DFA input size.
Larger DFA inputs incur longer overall chunk matching times, which
de-emphasize the high (but constant) communication costs for a given DFA.
Fig.~\ref{fig:mpi_comm_cost}(b) depicts the proportional communication costs
that we measured for varying input sizes. It follows that for \SI{1}{GB} and
\SI{10}{GB}
input sizes, the proportion of communication overhead with respect to the
overall execution time is close to zero. However, with input sizes of \SI{100}{MB} and
especially \SI{10}{MB}, the time spent for communication constitutes a large part of
the overall execution time.  As shown in Fig.~\ref{fig:mpi_comm_cost}(a),
our DFA matching approach scaled well for input-sizes of up to \SI{10}{GB} on the EC2
computing cloud. 

For sake of completeness, we state the execution time costs for computing DFAs
from PROSITE protein patterns. 
We used Grail+ to create nondeterministic finite automata (NFAs) from regular expressions,
convert NFAs to DFAs, and minimize DFAs. We did not apply parallel versions
of algorithms for DFA creation and minimization~\cite{psubset,Tewari:2002}.  
On average, it took \SI{8}{min} \SI{21.5}{s} to convert a PROSITE pattern to a minimal
DFA, and \SI{4}{min} \SI{35.8}{s} to convert a pattern to a non-minimal DFA. Our
speculative DFA matching approach is suitable to both minimal and non-minimal DFAs.
\section{Related Work}\label{sec:RelatedWork}
Locating a string in a larger text has applications
with text editing, compiler front-ends and web browsers,
Internet search engines, computer security, and DNA sequence analysis. 
Early string
searching algorithms such as Aho--Corasick~\cite{Aho:1975}, 
Boyer--Moore~\cite{Boyer:1977} and Rabin--Karp~\cite{Karp:1987}
efficiently match a finite set of input strings
against an input text. 

Regular expressions allow the specification of
infinite sets of input strings.
Converting a regular expression to a DFA for DFA membership tests is a 
standard technique to perform regular expression matching. 
The specification of virus signatures in intrusion
prevention systems~\cite{Brumley2006,Sommer2003,Roesch1999} and
the specification of DNA 
sequences~\cite{sigrist2010prosite,boeckmann2003swiss}
constitute recent applications of regular expression matching with DFAs.

Considerable research effort has been
spent on parallel algorithms for DFA membership tests.
Ladner et al.~\cite{Ladner:1980} applied the parallel prefix computation 
for DFA membership tests with Mealy machines.
Hillis and Steele~\cite{Hillis:1986} applied parallel prefix computations 
for DFA membership tests on the 65,536 processor Connection Machine.
Ravikumar's survey~\cite{ParallelFA1998} shows how DFA membership tests 
can be stated as a chained product of matrices.
Because of the underlying parallel prefix computation, all three approaches
perform a DFA membership test on input size~$n$ 
in $\BigO(\log(n))$ steps, requiring $n$~processors. Their algorithms 
handle arbitrary
regular expressions, but the underlying assumption of a massive number 
of available
processors can hardly be met in most practical settings. 
Misra~\cite{Misra:2003} derived another $\BigO(\log(n))$ string
matching algorithm. The number of required processors is
on the order of the product of the two string lengths and hence not 
practical.

A straight-forward way to exploit parallelism with DFA membership tests is
to run a single DFA on multiple input streams in parallel, or to run 
multiple DFAs in parallel. This approach has been taken by
Scarpazza et al.~\cite{ScarpazzaVP07} with a DFA-based string matching 
system for network security on the IBM Cell BE processor.
Similarly, Wang et
al.~\cite{WangHL10} investigated parallel architectures for packet
inspection based on DFAs. Both approaches assume multiple input streams 
and a vast number of
patterns (i.e., virus signatures), which is common with network security 
applications. However, neither
approach parallelizes the DFA membership algorithm itself, which is 
required to improve
applications with single, long-running membership tests such as DNA 
sequence analysis.

Scarpazza et al.~\cite{ScarpazzaVP07} utilize the SIMD units of the Cell BE's synergistic
processing units to match multiple input streams in parallel. However, 
their vectorized DFA matching
algorithm contains several SISD instructions and the reported speedup 
from 16-way vectorization is only a factor of 2.51.
In contrast, our proposed 8-way vectorized DFA membership test avoids SISD 
instructions, achieving a speedup of 4.45 over the sequential version.

Recent research efforts focused on speculative computations to 
parallelize DFA membership tests.  
Holub and \v{S}tekr~\cite{Holub:2009} were the first to split the input
string into chunks and distribute chunks among available processors. Their
speculation introduces a substantial amount of redundant computation, which
restricts the obtainable speedup for general DFAs
to~$\BigO(\frac{\lvert P\rvert}{\lvert Q\rvert})$, where 
$\lvert P\rvert$ is
the number of processors, and $\lvert Q\rvert$ is the number of DFA
states. 
Their algorithm degenerates to a speed-down when
$|\States|$ exceeds the number of processors (see also 
Section~\ref{sec:ExperimentalResults},
Fig.~\ref{fig:speedup_mtl_holub}).
To overcome this problem, Holub and \v{S}tekr specialized their 
algorithm for
$k$-local DFAs. A DFA is $k$-local if for every word of length~$k$ and
for all states~$p,q\in Q$ it holds that
$\TFE(p, w) = \TFE(q,w)$. Starting the matching operation $k$ symbols 
ahead of a given
chunk will synchronize the DFA into the correct initial state by 
the time matching reaches the beginning of the chunk, which eliminates 
all speculative computation. Holub and  \v{S}tekr achieve
a linear speedup of $\BigO(\lvert P\rvert)$ for $k$-local automata.
Unlike Holub and \v{S}tekr's approach, our DFA parallelization avoids 
speed-downs altogether.
We use structural properties of general DFAs to limit the amount of 
speculation. In particular,
the restriction to $k$-local automata is not required. We have vectorized 
our speculative
matching routine, and we have extensively evaluated our approach on a 
40-core shared memory
architecture, for AVX2 vector instructions, and on the Amazon EC2 cloud 
infrastructure. 

Jones et al.~\cite{Jones2009} reported that with the IE~8 and Firefox 
web browsers
3--40\% of the execution-time is spent parsing HTML documents. To speed 
up browsing, Jones et al.\
employ speculation to parallelize token detection (lexing) of HTML 
language front-ends.
Similar to Holub and \v{S}tekr's $k$-local automata, they use the 
preceding~$k$
characters of a chunk to synchronize a DFA to a particular state. 
Unlike $k$-locality,
which is a static DFA property, Jones et al.\
speculate the DFA to be in a particular, frequently occurring DFA state 
at the beginning of a chunk. Speculation fails if the DFA turns 
out to be in a different state, in which case the chunk
needs to be re-matched. Lexing HTML documents results in frequent 
matches, and the structure of regular expressions is reported to be 
simpler than, e.g., virus signatures~\cite{Luchaup2011}. Speculation is 
facilitated by the fact that the state at the beginning
of a token is always the same, regardless where lexing started. A 
prototype implementation is reported to scale up to six of the eight 
synergistic processing units of the Cell BE.

The speculative parallel pattern matching (SPPM) approach by
Luchaup et al.~\cite{Luchaup2011,Luchaup2009} uses speculation to match 
the increasing network line-speeds
faced by intrusion prevention systems. SPPM DFAs represent virus 
signatures. Like Jones et al.,
DFAs are speculated to be in a particular, frequently occurring DFA state 
at the beginning of a chunk. SPPM starts the speculative matching 
at the beginning of each chunk. With every input character, a 
speculative matching process stores the encountered DFA
state for subsequent reference. Speculation fails if the DFA turns out 
to be in a different state at the beginning of a speculatively matched 
chunk. In this case re-matching continues until the DFA 
synchronizes with the saved history state (in the worst case, the
whole chunk needs to be re-matched). A single-threaded SPPM 
version is proposed to improve performance by issuing multiple 
independent memory accesses in parallel.
Such pipelining (or interleaving) of DFA matches is orthogonal to 
our approach, which
focuses on latency rather than throughput.

SPPM assumes all regular expressions to be suffix-closed, which is the 
common scenario with intrusion prevention systems; 
A regular expression is suffix-closed if matching a given string~$w$ 
implies that $w$ followed by any suffix is matched, too. A suffix-closed 
regular language has the property that 
$x\in L \Leftrightarrow\forall w \in\Sigma^*: xw\in L$.

Unlike SPPM and the approach by Jones et al., our speculative DFA 
matching approach does not rely on a heavily biased distribution of DFA 
state frequencies. Instead, we use static DFA properties to minimize 
speculative matching overhead. Our approach is not restricted to 
suffix-closed regular expressions, and our speculation does not rely on
the common case being a match (Jones et al.), or the common case being a
non-match (SPPM). To the best of our knowledge, we are the first to employ
SIMD gather-operations to fully vectorized the DFA matching process.
Our DFA membership test provides a load-balancing mechanism
for clusters and cloud computing
environments. Unlike previous approaches, our speculative matching 
algorithm cannot result in a speed-down. We
conducted an extensive experimental evaluation on a 40-core shared memory
architecture, on a simulator for AVX2 vector instructions,
and on the Amazon EC2 cloud infrastructure. Our benchmarks consist of
\NrPCREs\ regular expressions from
the PCRE library~\cite{PCRELib}, and
of \NrPROSITEs\ patterns from the PROSITE protein pattern
database~\cite{sigrist2010prosite}. 
We analyzed the complexity of our speculative matching algorithm, and we 
provide insight on achievable scalability on shared-memory and 
cloud-computing environments.
This paper is the extended, journal version of an informal one-page
abstract
presented at the 4th annual meeting of the Asian Association
for Algorithms and Computation~\cite{aaacweb},
and a preliminary technical report~\cite{PDFA}.

\section{Conclusions}\label{sec:Conclusion}
We have presented a speculative DFA pattern matching method
for shared-memory, SIMD and cloud computing environments.
Our parallel
matching algorithm exploits structural DFA properties to minimize
the speculative overhead.
To the best of our knowledge, this is the first speculative
DFA matching approach that is {\em failure-free\/}, i.e., (1)~it maintains
sequential
semantics, and (2)~it avoids speed-downs altogether.
On architectures with a SIMD gather-operation for indexed memory loads,
our matching operation is fully vectorized. Communication patterns specifically
for the characteristics of cloud computing environments are provided. 
The proposed load-balancing scheme uses
an off-line profiling step to determine the matching capacity of each
participating processor. Based on matching capacities, DFA matches
are load-balanced on inhomogeneous parallel architectures.
We have shown that our algorithms have a better
time complexity than previous work.
We conducted an extensive experimental evaluation
of PCRE and PROSITE benchmarks on
a 4~CPU (40~cores) shared-memory node of the Intel
Academic Program
Manycore Testing Lab (Intel MTL),
on the Intel AVX2 SDE simulator for 8-way
fully vectorized SIMD execution, and on a 20-node (288~cores)
cluster of the Amazon EC2 computing cloud.
We showed the scalability of our approach for DFAs of up to 1288~states, and
input-strings of up to \SI{10}{GB}.
Our results predict that speculative parallel DFA matching
can produce substantial speedups. Unlike previous methods, our technique
does not impose any restriction on the matched regular expressions.

\begin{acknowledgements}
Research partially supported by the National Research Foundation of 
Korea~(NRF) grants
funded by the Korean government~(MEST)~(grant no.~2010-0005234, 2012R1A1A2044562 and 
2012K2A1A9054713), through the Global Ph.D.\
Fellowship Program~2011 of the NRF (grant no.~2010-0008582),
and by the Intel Academic Program Manycore Testing Lab.
\end{acknowledgements}
\bibliographystyle{spmpsci}
\bibliography{paper}

\begin{thebibliography}{10}
\providecommand{\url}[1]{{#1}}
\providecommand{\urlprefix}{URL }
\expandafter\ifx\csname urlstyle\endcsname\relax
  \providecommand{\doi}[1]{DOI~\discretionary{}{}{}#1}\else
  \providecommand{\doi}{DOI~\discretionary{}{}{}\begingroup
  \urlstyle{rm}\Url}\fi

\bibitem{MPICH2}
{A high-performance and widely portable implementation of the MPI standard
  (MPICH2) Web Site}: \url{http://www.mcs.anl.gov/research/projects/mpich2}
  (retrieved Aug.~2012).
\newblock MPICH2 version 1.4

\bibitem{Aho:1975}
Aho, A.V., Corasick, M.J.: Efficient string matching: an aid to bibliographic
  search.
\newblock Commun. ACM \textbf{18}(6), 333--340 (1975)

\bibitem{EC2}
{Amazon Web Services}: {EC2} web site, \url{http://aws.amazon.com/ec2}
  (retrieved Aug.~2012)

\bibitem{EC2limit}
{Amazon Web Services}: {EC2} {FAQs}, \url{http://aws.amazon.com/ec2/faqs}
  (retrieved Jan.~2013)

\bibitem{Armbrust2009}
Armbrust, M., Fox, A., Griffith, R., Joseph, A.D., Katz, R.H., Konwinski, A.,
  Lee, G., Patterson, D.A., Rabkin, A., Zaharia, M.: Above the clouds: A
  {Berkeley} view of cloud computing.
\newblock Tech. rep., University of California at Berkeley, Electrical
  Engineering and Computer Sciences (2009)

\bibitem{aaacweb}
{Asian Association for Algorithms and Computation (AAAC)}: 4th annual meeting
  website (2011), \url{http://3glab.cs.nthu.edu.tw/aaac2011} (retrieved
  Jan.~2013)

\bibitem{PDFA}
{Bernd Burgstaller}, {Yo-Sub Han}, {Minyoung Jung}, {Yousun Ko}: On the
  parallelization of {DFA} membership tests.
\newblock Tech. Rep. TR-0003, Dept.~Computer Science, Yonsei University, Seoul
  120-749, Korea, \url{http://elc.yonsei.ac.kr/PDFA.html} (2011)

\bibitem{boeckmann2003swiss}
Boeckmann, B., Bairoch, A., Apweiler, R., Blatter, M., Estreicher, A.,
  Gasteiger, E., Martin, M., Michoud, K., O'Donovan, C., Phan, I., et~al.: {The
  SWISS-PROT protein knowledgebase}.
\newblock Nucleic acids research \textbf{31}(1), 365 (2003)

\bibitem{Boyer:1977}
Boyer, R.S., Moore, J.S.: A fast string searching algorithm.
\newblock Commun. ACM \textbf{20}(10), 762--772 (1977)

\bibitem{Brumley2006}
Brumley, D., Newsome, J., Song, D., Wang, H., Jha, S.: Towards automatic
  generation of vulnerability-based signatures.
\newblock In: Proceedings of the 2006 IEEE Symposium on Security and Privacy,
  SP '06, pp. 2--16. IEEE Computer Society (2006).
\newblock \doi{10.1109/SP.2006.41}

\bibitem{Butenhof97}
Butenhof, D.R.: Programming with POSIX threads.
\newblock Addison-Wesley Longman Publishing Co., Inc., Boston, MA, USA (1997)

\bibitem{psubset}
Choi, H., Burgstaller, B.: Non-blocking parallel subset construction on
  shared-memory multicore architectures.
\newblock In: Proceedings of the 11th Australasian Symposium on Parallel and
  Distributed Computing (AusPDC 2013). CRPIT (2013)

\bibitem{cox2007}
Cox, R.: {Regular expression matching can be simple and fast (but is slow in
  {Java}, {Perl}, {PHP}, {Python}, {Ruby},...)} (2007).
\newblock \urlprefix\url{http://swtch.com/~rsc/regexp/regexp1.html}

\bibitem{ScanProsite:2002}
Gattiker, A., Gasteiger, E., Bairoch, A.: {ScanProsite}: a reference
  implementation of a {PROSITE} scanning tool.
\newblock Applied Bioinformatics \textbf{1}(2) (2002)

\bibitem{Grail}
{Grail+ Project Web Site}: \url{http://www.csd.uwo.ca/Research/grail}
  (retrieved Aug.~2012)

\bibitem{CellAccelerators}
Gschwind, M., Hofstee, H.P., Flachs, B., Hopkins, M., Watanabe, Y., Yamazaki,
  T.: {Synergistic Processing in Cell's Multicore Architecture}.
\newblock IEEE Micro \textbf{26}(2), 10--24 (2006).
\newblock \doi{http://dx.doi.org/10.1109/MM.2006.41}

\bibitem{haertel2010}
Haertel, M.: {Why GNU grep is fast} (2010).
\newblock
  \urlprefix\url{http://lists.freebsd.org/pipermail/freebsd-current/2010-Augus%
t/019310.html}

\bibitem{Hillis:1986}
Hillis, W.D., Steele Jr., G.L.: Data parallel algorithms.
\newblock Commun. ACM \textbf{29}(12), 1170--1183 (1986).
\newblock \doi{10.1145/7902.7903}

\bibitem{Holub:2009}
Holub, J., \v{S}tekr, S.: On parallel implementations of deterministic finite
  automata.
\newblock In: Proceedings of the 14th International Conference on
  Implementation and Application of Automata, pp. 54--64 (2009)

\bibitem{MTL}
{Intel Academic Program Manycore Testing Lab Site}:
  \url{http://software.intel.com/en-us/articles/intel-many-core-testing-lab}
  (retrieved Aug.~2012)

\bibitem{AVX2}
{Intel Advanced Vector Extensions Programming Reference}:
  \url{http://software.intel.com/en-us/avx} (retrieved Aug.~2012).
\newblock JUNE 2011 version

\bibitem{SDE}
{Intel Software Development Emulator (SDE) Web Site}:
  \url{http://software.intel.com/en-us/articles/intel-software-development-emu%
lator} (retrieved Aug.~2012).
\newblock SDE version 4.46.0

\bibitem{Jones2009}
Jones, C.G., Liu, R., Meyerovich, L., Asanovi\'{c}, K., Bod\'{\i}k, R.:
  Parallelizing the web browser.
\newblock In: Proceedings of the First {USENIX} conference on Hot topics in
  parallelism, HotPar'09, pp. 7--7. USENIX Association, Berkeley, CA, USA
  (2009)

\bibitem{Karonis:2000}
Karonis, N., Supinski, B.R.D., Foster, I., Gropp, W., Lusk, E., Bresnahan, J.:
  Exploiting hierarchy in parallel computer networks to optimize collective
  operation performance.
\newblock In: Proceedings of the 14th International Symposium on Parallel and
  Distributed Processing, IPDPS '00, pp. 377--. IEEE Computer Society,
  Washington, DC, USA (2000)

\bibitem{Karp:1987}
Karp, R.M., Rabin, M.O.: Efficient randomized pattern-matching algorithms.
\newblock IBM J. Res. Dev. \textbf{31}(2), 249--260 (1987)

\bibitem{Ladner:1980}
Ladner, R.E., Fischer, M.J.: Parallel prefix computation.
\newblock J. ACM \textbf{27}(4), 831--838 (1980).
\newblock \doi{10.1145/322217.322232}

\bibitem{Snyder}
Lin, C., Snyder, L.: Principles of Parallel Programming.
\newblock Addison Wesley (2008)

\bibitem{Luchaup2009}
Luchaup, D., Smith, R., Estan, C., Jha, S.: Multi-byte regular expression
  matching with speculation.
\newblock In: Proceedings of the 12th International Symposium on Recent
  Advances in Intrusion Detection, RAID '09, pp. 284--303. Springer-Verlag,
  Berlin, Heidelberg (2009).
\newblock \doi{10.1007/978-3-642-04342-0_15}

\bibitem{Luchaup2011}
Luchaup, D., Smith, R., Estan, C., Jha, S.: Speculative parallel pattern
  matching.
\newblock IEEE Transactions on Information Forensics and Security
  \textbf{6}(2), 438--451 (2011)

\bibitem{Specpar}
Luj{\'a}n, M., Gustafson, P., Paleczny, M., Vick, C.A.: Speculative
  parallelization---eliminating the overhead of failure.
\newblock In: R.H. Perrott, B.M. Chapman, J.~Subhlok, R.F. de~Mello, L.T. Yang
  (eds.) HPCC, \emph{Lecture Notes in Computer Science}, vol. 4782, pp.
  460--471. Springer (2007)

\bibitem{Misra:2003}
Misra, J.: Derivation of a parallel string matching algorithm.
\newblock Inf. Process. Lett. \textbf{85}(5), 255--260 (2003).
\newblock \doi{10.1016/S0020-0190(02)00416-7}

\bibitem{OpenMPI}
{OpenMPI Web Site}: \url{http://www.open-mpi.org} (retrieved Aug.~2012)

\bibitem{Ostermann09}
Ostermann, S., Iosup, A., Yigitbasi, N., Prodan, R., Fahringer, T., Epema, D.:
  {A Performance Analysis of {EC2} Cloud Computing Services for Scientific
  Computing}.
\newblock In: Cloud Computing, \emph{Lecture Notes of the Institute for
  Computer Sciences, Social Informatics and Telecommunications Engineering},
  vol.~34, chap.~9, pp. 115--131. Springer Berlin Heidelberg (2010)

\bibitem{PCRELib}
{Perl Compatible Regular Expression Library Web Site}:
  \url{http://www.pcre.org} (retrieved Aug.~2012)

\bibitem{PROSITE}
{PROSITE Web Site}: \url{http://prosite.expasy.org} (retrieved Aug.~2012)

\bibitem{ParallelFA1998}
Ravikumar, B.: Parallel algorithms for finite automata problems.
\newblock In: IPPS/SPDP Workshops, vol. 1388. Springer (1998)

\bibitem{Grail95}
Raymond, D., Wood, D.: Grail: A {C++} library for automata and expressions.
\newblock Journal of Symbolic Computation \textbf{17}, 17--341 (1995)

\bibitem{Roesch1999}
Roesch, M.: {Snort---Lightweight Intrusion Detection for Networks}.
\newblock In: Proceedings of the 13th USENIX conference on System
  administration, LISA '99, pp. 229--238. USENIX Association (1999)

\bibitem{ScanProsite}
{ScanProsite Web Site}: \url{http://prosite.expasy.org/scanprosite} (retrieved
  Aug.~2012)

\bibitem{ScarpazzaVP07}
Scarpazza, D.P., Villa, O., Petrini, F.: Peak-performance {DFA}-based string
  matching on the {Cell} processor.
\newblock In: 21th International Parallel and Distributed Processing Symposium,
  pp. 1--8 (2007)

\bibitem{EC2Perf}
Schad, J., Dittrich, J., Quian{\'e}-Ruiz, J.A.: Runtime measurements in the
  cloud: observing, analyzing, and reducing variance.
\newblock Proc. VLDB Endow. \textbf{3}(1-2), 460--471 (2010)

\bibitem{sigrist2010prosite}
Sigrist, C., Cerutti, L., De~Castro, E., Langendijk-Genevaux, P., Bulliard, V.,
  Bairoch, A., Hulo, N.: {PROSITE, a protein domain database for functional
  characterization and annotation}.
\newblock Nucleic acids research \textbf{38}(suppl 1), D161 (2010)

\bibitem{Sommer2003}
Sommer, R., Paxson, V.: Enhancing byte-level network intrusion detection
  signatures with context.
\newblock In: Proceedings of the 10th ACM conference on Computer and
  communications security, CCS '03, pp. 262--271. ACM (2003).
\newblock \doi{10.1145/948109.948145}

\bibitem{StarCluster}
{StarCluster cluster computing toolkit}: \url{http://web.mit.edu/star/cluster}
  (version 0.93.3, retrieved July~2012)

\bibitem{Tewari:2002}
Tewari, A., Srivastava, U., Gupta, P.: A parallel {DFA} minimization algorithm.
\newblock In: Proceedings of the 9th International Conference on High
  Performance Computing, HiPC '02, pp. 34--40. Springer-Verlag (2002)

\bibitem{Wang:2010}
Wang, G., Ng, T.S.E.: The impact of virtualization on network performance of
  {Amazon} {EC2} data center.
\newblock In: Proceedings of the 29th conference on Information communications,
  INFOCOM'10, pp. 1163--1171. IEEE Press (2010)

\bibitem{WangHL10}
Wang, X., He, K., Liu, B.: Parallel architecture for high throughput
  {DFA}-based deep packet inspection.
\newblock In: 2010 IEEE International Conference on Communications, pp. 1--5
  (2010)

\end{thebibliography}
\end{document}